\documentclass[preprint,12pt]{article}

\usepackage{xcolor}
\usepackage{enumerate}
\usepackage{subfiles}
\usepackage{changepage}
\usepackage[font={small,it}]{caption}
\usepackage{verbatim}
\usepackage[english]{babel}
\usepackage[utf8]{inputenc}
\usepackage[margin=2cm]{geometry}
\usepackage{float}
\usepackage{times}
\usepackage{soul}
\usepackage{url}
\usepackage[hidelinks]{hyperref}
\usepackage[utf8]{inputenc}
\usepackage{graphicx}
\usepackage{amsmath}
\usepackage{amsfonts}
\usepackage{amsthm}
\usepackage{booktabs}
\usepackage{subcaption}
\theoremstyle{definition}
\newtheorem{definition}{Definition}
\newtheorem{theorem}{Theorem}
\newtheorem{proposition}{Proposition}
\newtheorem{lemma}{Lemma}
\newtheorem{corollary}{Corollary}

\theoremstyle{remark}
\newtheorem*{remark}{Remark}

\newcommand{\agentset}{\mathcal{N}}
\newcommand{\edgeset}{\mathcal{E}}

\newcommand{\NE}{\mathbf{\bar{x}}}

\newcommand{\x}{\mathbf{x}}
\newcommand{\y}{\mathbf{y}}
\newcommand{\z}{\mathbf{z}}
\newcommand{\e}{\mathbf{e}}
\newcommand{\p}{\mathbf{p}}

\newcommand{\R}{\mathbb{R}}

\usepackage{listings}

\usepackage{multirow}

\usepackage{fancyhdr}

\title{Beyond Strict Competition:\\ Approximate Convergence of Multi Agent Q-Learning Dynamics}
\author{
	Aamal Abbas Hussain\\
	\texttt{aamal.hussain15@imperial.ac.uk}
	\and
	Francesco Belardinelli\\
	\texttt{francesco.belardinelli@imperial.ac.uk}
	\and 
	Georgios Piliouras\\
	\texttt{georgios@sutd.edu.sg}
}
\begin{document}
\maketitle

    \begin{abstract}
    The behaviour of multi-agent learning in competitive settings is often considered under the restrictive assumption of a zero-sum game. Only under this strict requirement is the behaviour of learning well understood; beyond this, learning dynamics can often display non-convergent behaviours which prevent fixed-point analysis. Nonetheless, many relevant competitive games do not satisfy the zero-sum assumption.
    Motivated by this, we study a smooth variant of Q-Learning, a popular reinforcement learning dynamics which balances the agents' tendency to maximise their payoffs with their propensity to explore the state space. We examine this dynamic in games which are `close' to network zero-sum games and find that Q-Learning converges
    to a neighbourhood around a unique equilibrium. The size of the neighbourhood is determined by the `distance' to the zero-sum game, as well as the exploration rates of the agents. We complement these results by providing a method whereby, given an arbitrary network game, the `nearest' network zero-sum game can be found efficiently. As our experiments show, these guarantees are independent of whether the dynamics ultimately reach an equilibrium, or remain non-convergent. 
    \end{abstract}

    \section{Introduction} % (fold)
    \label{sec:introduction}
    
        The convergence of multi-agent learning in competitive settings has long been studied under the context of zero-sum games. The ability to make
        strong predictions in zero-sum games follows from its enforcement of strict competition between agents. Indeed many positive results have been
        achieved which show the convergence, in time average, of no regret learning algorithms to a Nash Equilibrium (NE) \cite{mertikopoulos:finite,bailey:fast-furious}. Yet time average convergence does not always imply convergence of the last-iterate. Under this context, zero-sum games, and their network variants, have received much attention, showing cyclic behaviour for some
         algorithms \cite{piliouras:cycles,hofbauer:hamiltonian} and asymptotic convergence for others 
        \cite{piliouras:zerosum,ewerhart:fp,hofbauer:zerosum}.
    
        Yet in multi-agent settings the satisfaction of strict competition cannot be taken for granted. The reason for this is simple: not all competitive
        games are zero-sum. Another contributor is noise; in practice payoffs measured by agents may be subject
        to perturbations so that the underlying game no longer satisfies the zero-sum condition. It is natural, then, to ask whether the convergence
        structure holds as we move away from the requirement of strict competition.
    
        Unfortunately, the general answer to this question is \emph{no}. Learning algorithms are known to display complex, even chaotic behaviour when
        even slightly perturbed away from the safe haven of zero sum games \cite{sato:rps,galla:complex,galla:cycles,cheung:decomposition}. In fact, this
        problem becomes even more prevalent as the number of players is increased \cite{sanders:chaos}. The introduction of chaos makes
        the exact prediction of long-term behaviours impossible in a wide class of games and we are led to a fundamental dichotomy between the need, and
        ability to understand multi-agent learning in competitive games. 
    
        \paragraph{Summary of Main Contributions} % (fold)
        \label{par:summary_of_main_contributions}
    
        To make progress in understanding general competitive games, we consider the natural starting point of \emph{near network zero-sum games}. The
        concept of `close' games has been introduced in the context of potential games \cite{candogan:nearpotential}, which model strictly cooperative
        settings. Following its introduction, a number of results on the approximate convergence of learning algorithms have been determined in
        near-potential games \cite{anagnostides:last-iterate,cheng:nwpg} Motivated by the success of the cooperative setting, we
        re-purpose the distance notion for network zero-sum games (NZSG) which form the natural extension of the zero-sum game to multi
        agent settings \cite{cai:minimax}.
        
        In this setting, we study the \emph{(smooth) Q-Learning dynamics} which models the popular Q-Learning algorithm with Boltzmann exploration 
        \cite{sutton:barto,schwartz:MARL}. This learning model captures the behaviour of agents who attempt to maximise their payoffs whilst balancing a
        tendency to explore the space of their possible strategies.
    
        Our first contribution is to show that, in near network zero-sum games, Q-Learning converges to a neighbourhood around the unique
        equilibrium of the underlying NZSG. The size of this set goes to zero as the distance from the NZSG goes to zero and/or as the exploration rate
        of each agent increases. Given, then, the distance from the NZSG this size of the neighbourhood can be adjusted by manipulating the exploration
        rates of the agents. To assist in this process, we also provide upper bounds on the distance between network games based on the differences in
        payoff matrices and the network structure. Finally, in a similar light to \cite{candogan:nearpotential,cheng:nwpg} which consider potential games,
        we present a quadratic optimisation formulation for determining the closest NZSG to a given network game. Taken together, these results give a
        picture of the approximate behaviour of Q-Learning in competitive games which do not exactly satisfy the zero-sum condition. 
    
        % paragraph summary_of_main_contributions (end)
    
    \subsection{Related Work} % (fold)
    \label{sub:related_work}
        Studies on learning in competitive games often occur within the context of zero-sum games \cite{aumann:gt} or its network variants \cite{cai:minimax}. Indeed, due to the desirable structure of these games and the increasing interest of competitive systems \cite{abernethy:hamiltonian}, many positive results have been obtained concerning various learning dynamics, including Follow the Regularised Leader \cite{piliouras:hamiltonian,anagnostides:last-iterate}, fictitious play \cite{ewerhart:fp}, and Q-Learning \cite{piliouras:zerosum}. 
    
        By contrast, little is known about the behaviour of learning once we leave these games. In fact, non-convergent behaviour, including cycles and chaos, 
        appears to be increasingly prevalent as the NZSG condition is lifted \cite{galla:cycles,sato:rps,sato:qlearning,chakraborty:chaos} and as the number of agents increases \cite{sanders:chaos}. This presents a strong barrier when attempting to engineer competitive multi-agent systems, where the network zero-sum assumption need not hold \cite{ewerhart:fp,roberson:blotto}. Outside of this class, results on convergence often make
        restrictive assumptions, such as the existence of a potential function \cite{piliouras:potential,Monderer1996PotentialGames,harris:fp} which enforces strict cooperation amongst agents, or that the game has only two players and two actions \cite{galstyan:2x2,metrick:fp}. Of course,
        these do not cover the vast majority of games encountered in practice. In fact, the strongest result regarding learning outside of NZSG is a negative one: consider \cite{flokas:donotmix} which shows that the popular Follow the Regularised Leader dynamic cannot
        converge to a fully mixed Nash Equilibrium, regardless of the game structure. With all these taken together, it becomes clear that a complete
        picture of learning in games cannot be found by considering only convergence to a fixed point, but must include the eventuality of non-convergence.
    
        To make progress on this, we apply the concept of `nearness' in games. This was first introduced in the context of potential games \cite{candogan:nearpotential,cheng:nwpg} to extend the analysis of cooperative games to those which do not satisfy the
        potential assumption. With this, various learning algorithms including fictitious play \cite{candogan:nearpotential,aydin:fp} and Follow
        the Regularised Leader \cite{anagnostides:last-iterate}, could be understood in terms of \emph{approximate convergence}, i.e. convergence
        to a neighbourhood of an equilibrium. On the other hand, whilst \cite{cheung:decomposition} shows that games which deviate from the network zero sum setting can display chaos, little is known about how deviations from the strictly competitive setting affect the approximate convergence of learning. To our knowledge, the present work is the first to study, both theoretically and experimentally, near network zero sum games with an aim to understand approximate convergence, even in the face of chaos.
    
    % subsection related_work (end)
    
    % section introduction (end)
    
    \section{Preliminaries} % (fold)
    \label{sec:preliminaries}
    
    We study a game $\Gamma = (\agentset, (S_k, u_k)_{k \in \agentset})$ where $\agentset$ denotes
    a finite set of
    agents indexed by $k = 1, \ldots, N$. Each agent $k \in \agentset$ has a finite set of actions $S_k$ which are indexed
    by $i = 1, \ldots, n_k$. Players can also play a mixed strategy $\x_k$ which is a discrete probability distribution over
    its set of actions. The set of all such mixed strategies is the unit simplex in $\R^{n_k}$. More formally, the simplex
    associated to agent $k$ is $\Delta_k = \{ \x_k \in \R^{n_k} \, : \, \sum_{i \in S_k} x_{ki} = 1, \, x_{ki} \geq
    0, \, \forall i \in S_k\}$. We denote $\Delta = \times_{k \in \agentset} \Delta_k$ as the joint simplex over all
    agents, $\x = (\x_k)_{k \in \agentset}$ as the joint mixed strategy of all agents and, for any $k$, $\x_{-k} = (\x_l)_
    {l \in \agentset\backslash\{k\}} \in \Delta_{-k}$ as the joint strategy of all agents other than $k$.
    
    Also associated to each agent $k$ is a payoff function $u_k\,: \, \Delta_k \times \Delta_{-k} \rightarrow \R$. Then, for
    any $\x \in \Delta$, we define the reward to agent $k$ when they play action $i \in S_k$ as $r_{ki}(\x) := 
    \frac{\partial u_{ki}(\x)}{\partial x_{ki}}$. With this, we can write $r_k(\x) = (r_{ki}(\x))_{k \in \agentset}$ as the
    concatenation of all rewards to agent $k$. In this notation, $u_k(\x) = \langle \x_k, r_k(\x) \rangle$ where $\langle
    \x, \y \rangle = \x^\top \y$ is the inner product in $\R^n$. 
    
    \paragraph{Network Zero-Sum Games} % (fold)
    \label{par:network_zero_sum_games} A \emph{polymatrix} or \emph{network} game also contains a graph $(\agentset,
    \edgeset)$ in which $\agentset$ still denotes the set of agents and $\edgeset$ consists of pairs $(k, l) \in
    \agentset$ of agents, who are meant to be connected \cite{cai:minimax}.
    Each edge has associated a pair  $(A^{kl}, A^{lk})$ of matrices, which define the payoff to
    $k$ against $l$ and vice versa. The payoffs are then given by
    \begin{equation*}
        u_k(\x_k, \x_{-k}) = \sum_{(k, l) \in \edgeset} \langle \x_k, A^{kl} \x_l \rangle
    \end{equation*}

    We represent a network game as a tuple $\Gamma = (\agentset, \edgeset, (S_k)_{k \in \agentset}, (A^{kl}, A^{lk})_{(k,
    l) \in \edgeset})$. $\Gamma$ is a \emph{network zero-sum game} (NZSG) if, for all $\x \in \Delta$,
    \begin{equation*} \label{eqn::NZSG}
        \sum_k u_k(\x_k, \x_{-k}) = 0
    \end{equation*}
    
    A seminal result in the study of NZSG is that of \cite{cai:minimax}
    which shows that any NZSG is payoff equivalent to a pairwise constant sum game, where all the constants add to zero. More formally, this is stated in the following proposition
    
    \begin{proposition}[\cite{cai:minimax}, \cite{piliouras:zerosum}] \label{prop::cai-nzsg}
        Let $Z = (\agentset, \edgeset, (S_k)_{k \in \agentset}, (A^{kl}, A^{lk})_{(k,l) \in \edgeset})$ be a NZSG. For all $(k, l) \in \edgeset$
        there exist $(\hat{A}^{kl}, \hat{A}^{lk})$ and a constant $c_{kl} \in \R$ such that
        \begin{equation*}
            [\hat{A}^{kl}]_{ij} + [\hat{A}^{lk}]_{ji} = c_{kl}, \; \forall i \in S_k, j \in S_l,
        \end{equation*}
        with
        \begin{equation*}
            \sum_{(k, l) \in \edgeset} c_{kl} = 0,
        \end{equation*}
        and payoffs to agent $k$ in $Z$ is equivalent to their payoffs in $\hat{Z} = (\agentset, \edgeset, (S_k)_{k \in \agentset}, (\hat{A}^{kl}, \hat{A}^{lk})_{(k,l) \in \edgeset})$. In particular, for all $k \in \agentset$ and all $\x_k \in \Delta_k$
        \begin{equation*}
            \sum_{(k, l) \in \edgeset} \x_k^\top \hat{A}^{kl} \x_l = \sum_{(k, l) \in \edgeset} \x_k^\top A^{kl} \x_l
        \end{equation*}
    \end{proposition}
    %
    % paragraph network_zero_sum_games (end)
    
    \paragraph{Maximum Pairwise Difference} % (fold)
    \label{par:maximum_pairwise_difference}
    
    To define `nearness' in the context of games, we require a notion of distance on the space of all games. We apply the
    widely used metric defined in \cite{candogan:nearpotential}, known as Maximum Pairwise Difference. Formally, let
    $\Gamma_1 =
    (\agentset, (S_k, A_k)_{k \in \agentset})$ and $\Gamma_2 = (\agentset, (S_k, B_k)_{k \in \agentset})$ be two games which share the same set of agents $\agentset$ and actionsets $(S_k)_{k \in \agentset}$ but differ in payoff functions. Then, the Maximum Pairwise
    Difference between $\Gamma_1$ and $\Gamma_2$ is
    \begin{equation} \tag{MPD} \label{eqn::MPD}
        d(\Gamma_1, \Gamma_2) = \max |A_k(\y_k, \x_{-k}) - A_k(\x_k, \x_{-k}) - (B_k(\y_k, \x_{-k}) - B_k(\x_k, \x_{-k}))|
    \end{equation}
    where the maximum is taken over all agents $k$, all $\x_{-k} \in \Delta_{-k}$ and all $\x_k, \y_k \in \Delta_k$. In words, (\ref{eqn::MPD}) captures the similarity between two games in terms of the capacity for any agent to improve their payoff by deviating from $\x_k$ to $\y_k$ whilst their opponents maintain their strategy $x_{-k}$.

    % paragraph maximum_pairwise_difference (end)
    
    \paragraph{Q-Learning Dynamics} % (fold)
    \label{par:q_learning_dynamics}
    
    Q-Learning is the prototypical model for determining optimal policies in the face of uncertainty. In this model, each
    agent $k \in \agentset$ maintains a history of the past performance of each of their actions. This history is updated
    via the Q-update
    \begin{equation*}
        Q_{ki}(\tau + 1) = (1 - \alpha_k) Q_{ki}(\tau) + \alpha_k r_{ki}(\x_{-k}(\tau))
    \end{equation*}
    where $\tau$ denotes the current time step. $Q_{ki}(\tau)$ dentoes the \emph{Q-value} maintained by agent $k$ about the
    performance of action $i \in S_k$. In effect $Q_{ki}$ gives a discounted history of the rewards received when $i$ is
    played, with $1 - \alpha_k$ as the discount factor.
    
    Given these Q-values, each agent updates their mixed strategies according to the Boltzmann distribution, given by
    \begin{equation*}
        x_{ki}(\tau) = \frac{\exp(Q_{ki}(\tau)/T_k) }{\sum_j \exp(Q_{kj}(\tau)/T_k)}
    \end{equation*}
    in which $T_k \in [0, \infty)$ is the \emph{exploration rate} of agent $k$: low values of $T_k$ allow the agent to
    play the action(s) with the highest Q-value with a large probability, thereby exploiting their high performance. By
    contrast, higher values of $T_k$ enforce that agents play each of their strategies with a high probability, regardless
    of their Q-value.
    
    It was shown in \cite{tuyls:qlearning,sato:qlearning} that a continuous time approximation of the Q-Learning algorithm
    could be written as
    \begin{equation} \tag{QLD} \label{eqn::QLD}
        \frac{\dot{x}_{k i}}{x_{k i}}=r_{k i}\left(\mathbf{x}_{-k}\right)-\langle \mathbf{x}_k, r_k(\mathbf{x}) \rangle +T_k \sum_{j \in S_k} x_{k j} \ln \frac{x_{k j}}{x_{k i}}
    \end{equation}
    which we call the \emph{Q-Learning dynamics}. The fixed points of this dynamic coincide with the \emph{Quantal Response
    Equilibria} (QRE) of the game.
    
    \begin{definition}[Quantal Response Equilibrium (QRE)]
        A joint mixed strategy $\p \in \Delta$ is a Quantal Response Equilibrium of the game $\Gamma = (\agentset, (S_k,
        u_k)_{k \in \agentset})$ if, for all agents $k \in \agentset$, $i \in S_k$ 
        \begin{equation} \label{eqn::QRE}
         p_{ki} = \frac{\exp(r_{ki}(\p_{-k})/T_k)}{\sum_{j \in S_k} \exp(r_{kj}(\p_{-k})/T_k)}   
        \end{equation}
    \end{definition}
    
    The QRE is a well studied equilibrium concept for games of \emph{bounded rationality} \cite{mckelvey:qre}. This is seen
    in
    the fact that, in the limit $T_k \rightarrow 0$ for all $k$, (\ref{eqn::QRE}) corresponds exactly to the Nash
    Equilibrium, whereas in the limit $T_k \rightarrow \infty$ for all $k$, the QRE is the uniform distribution, i.e. each
    agent plays each action with the same probability, regardless of its past performance.
    % paragraph q_learning_dynamics (end)
    
    \paragraph{Game Perturbations} % (fold)
    \label{par:game_perturbations} In \cite{piliouras:potential} it is shown that, for any $(T_k)_{k \in \agentset}$, the
    Q-Learning dynamics in a game $\Gamma$ is equivalent to the well studied \emph{replicator dynamics} (RD) in a perturbed
    game $\Gamma^H$. More formally, the authors show the following.
    
    \begin{lemma}[\cite{piliouras:potential}] \label{lem::QLRD}
        Consider a game $\Gamma = (\agentset, (S_k, u_k)_{k \in \agentset})$ and, for each agent $k$ let $T_k > 0$. Then (
        \ref{eqn::QLD}) can be written as
        \begin{equation}
            \frac{\dot{x}_{ki}}{x_{ki}} = r_{ki}^H(\x) - \langle \x_k, r_k^{H}(\x) \rangle,
        \end{equation}
        where $r_{ki}^H = r_{ki}(\x_{-k}) - T_k(\ln x_{ki} + 1)$. In particular, (\ref{eqn::QLD}) recovers the replicator
        dynamics in the pertubed game $\Gamma^H = (\agentset, (S_k, u^H_k)_{k \in \agentset})$ where
        \begin{equation*}
            u_k^H(\x) = \langle x_k, r_k(\x_{-k})\rangle - T_k \langle \x_k, \ln \x_k \rangle
        \end{equation*}
    \end{lemma}
    
    The perturbed game $\Gamma^H$ has the same players and action sets as $\Gamma$ but has modified utilities. The same
    perturbation maps the QRE of the game $\Gamma$ to Nash Equilibria of $\Gamma^H$ \cite{melo:qre,gemp:sample}
    
    \paragraph{Results on Game Perturbations.} % (fold)
    \label{par:results_on_game_perturbations}
    In the following, we consider how the definition of MPD relates to the perturbations between games. First, we show that the perturbation between games is distance preserving. 
    
    \begin{proposition} \label{prop::MPDperturbation}
         Let $\Gamma_1, \Gamma_2$ be games which share the same set of agents and action spaces, but differ in payoffs.
         Then,
         for any set of exploration rates $T_k > 0$, $d(\Gamma_1, \Gamma_2) = d(\Gamma^H_1, \Gamma^H_2)$.
    \end{proposition}
    \begin{proof}
    Let $\Gamma_1 = (\agentset, (S_k, A_k)_{k \in \agentset})$ and $\Gamma_2 = (\agentset, (S_k, B_k)_{k \in \agentset})$
    \begin{align*}
            d(\Gamma_1^H, \Gamma_2^H) =& \max |A^H_k(\y_k, \x_{-k}) - A^H_k(\x_k, \x_{-k}) - (B^H_k(\y_k, \x_{-k}) - B^H_k(\x_k, \x_{-k}))| \\
            =& \max |A_k(\y_k, \x_{-k}) - T_k \langle \y_k, \ln \y_k \rangle - A^H_k(\x_k, \x_{-k}) + T_k \langle \x_k, \ln \y_k \rangle \\ &- (B_k(\y_k,
            \x_{-k}) - T_k \langle \y_k, \ln \y_k \rangle  - B_k(\x_k, \x_{-k}) + T_k \langle \x_k, \ln \y_k \rangle )| \\
            =& \max |A_k(\y_k, \x_{-k}) - A_k(\x_k, \x_{-k}) - (B_k(\y_k, \x_{-k}) - B_k(\x_k, \x_{-k}))| = d(\Gamma_1, \Gamma_2)
    \end{align*}
    \end{proof}
    
    Next, we use MPD to show that any QRE is an approximate Nash Equilibrium, with the approximation parameterised
    by the exploration rates $T_1, \ldots, T_N$.
    
    \begin{proposition}
        For any game $\Gamma = (\agentset, (S_k, A_k)_{k \in \agentset})$ and any set of exploration rates $T_1, \ldots, T_N > 0$
        \begin{equation}
            d(\Gamma, \Gamma^H) \leq (\max_k T_k) (\max_k \ln n_k)
        \end{equation}
    \end{proposition}
    \begin{proof}
    \begin{align*}
           d(\Gamma, \Gamma^H) &= \max |u_k(\y_k, \x_{-k}) - u_k(\x_k, \x_{-k}) - (u^H_k(\y_k, \x_{-k}) - u^H_k(\x_k, \x_{-k}))| \\
           &= \max |u_k(\y_k, \x_{-k}) - u_k(\x_k, \x_{-k}) - (u_k(\y_k, \x_{-k}) - T_k \y_k^\top \ln \y_k - u_k(\x_k, \x_{-k}) + T_k \x_k^\top \ln \x_k)| \\
           &= \max |T_k \, (\x_k^\top \ln \x_k - \y_k^\top \ln \y_k)| \\
           &\leq (\max_k T_k) (\max_k \ln n_k)
    \end{align*}
    \end{proof}
    To show that a QRE is an approximate NE we apply the following.
    \begin{proposition}
        Let $\Gamma_1, \Gamma_2$ be games which share the same set of agents and action spaces, but differ in payoffs.
        Let $\delta = d(\Gamma_1, \Gamma_2)$. Then if $\NE \in \Delta$ is a Nash Equilibrium of $\Gamma_2$, it is a
        $\delta$-approximate Nash Equilibrium of $\Gamma_1$.
    \end{proposition}
    \begin{proof}
    Let $\Gamma_1 = (\agentset, (S_k, A_k)_{k \in \agentset})$ and $\Gamma_2 = (\agentset, (S_k, B_k)_{k \in \agentset})$. Then by the definition of $d(\Gamma_1, \Gamma_2)$, for all $k \in \agentset$, $\x_k, \y_k \in \Delta_k$, $\x_{-k} \in \Delta_{-k}$
        \begin{align*}
            |A_k(\y_k, \x_{-k}) - A_k(\x_k, \x_{-k}) - (B_k(\y_k, \x_{-k}) - B_k(\x_k, \x_{-k}))| &\leq \delta \\
        \end{align*}
    In particular, let us choose $\x_k, \x_{-k}$ as the Nash Equilibrium $\NE_k$ of $\Gamma_2$ so that $B_k(\y_k, \NE_{-k}) - B_k(\NE_k, \NE_{-k}) \leq 0$. Then, for all $\y_k$
        \begin{align*}
            A_k(\y_k, \NE_{-k}) - A_k(\NE_k, \NE_{-k}) &- (B_k(\y_k, \NE_{-k}) - B_k(\NE_k, \NE_{-k})) \leq \delta \\
            A_k(\y_k, \NE_{-k}) - A_k(\NE_k, \NE_{-k}) &\leq \delta + B_k(\y_k, \NE_{-k}) - B_k(\NE_k, \NE_{-k})\\
            A_k(\y_k, \NE_{-k}) - A_k(\NE_k, \NE_{-k}) &\leq \delta \\
            A_k(\y_k, \NE_{-k})  &\leq \delta + A_k(\NE_k, \NE_{-k})
        \end{align*}
    This forms the definition of a $\delta-$approximate NE in $\Gamma_1$
    \end{proof}
    Putting the above together immediately yields the following.
    \begin{corollary}
        For any game $\Gamma$ and any set of exploration rates $T_1, \ldots, T_N > 0$, the QRE of $\Gamma$ is a $(\max_k T_k) (\max_k \ln n_k)$-approximate Nash
        Equilibrium.
    \end{corollary}
    
    % paragraph results_on_game_perturbations (end)
    
    % subsubsection misc_results (end)

    \section{Near Network Zero-Sum Games} % (fold)
    \label{sec:near_network_zero_sum_games}
    
    Our main results concern the competitive setting. We first show that, in near-NZSG (\ref{eqn::QLD}) converges to a set
    around the QRE of the NZSG. This determines approximate convergence behaviour when the game is perturbed away from the
    NZSG assumption. We follow this with a scheme to determine, for any network game (not necessarily zero sum), the nearest
    NZSG. Using these results together provides a method to determine approximate convergence behaviour for arbitrary
    competitive network games.
    
    \subsection{Approximate Convergence} % (fold)
    \label{sub:approximate_convergence}
    
    To define the convergence of the Q-Learning dynamic we need a measure of distance. To this end, we use the \emph{Kullback-Leibler} (KL)
    divergence. 
    
    \begin{definition}
        The Kullback-Leibler Divergence between a set of joint strategies $\x, \y \in \Delta$ is given by
        \begin{equation}
            D_{KL}(\y || \x) = \sum_{k} D_{KL}(\y_k || \x_k) = \sum_k \sum_{i} y_{ki} \ln \frac{y_{ki}}{x_{ki}}
        \end{equation}
    \end{definition}
    
    Notice that the KL-Divergence does not formally define a metric as it is not symmetric (i.e. in general $D_{KL}(\y || \x) \neq 
    D_{KL}(\x || \y)$). Rather, the KL-Divergence can be thought of as measuring the overlap between probability distributions $\y$
    $\x$. The key point which we will use in our main theorem is that $D_{KL}(\y || \x)$ is zero if and only if $\x = \y$ and is
    positive everywhere else. 
    
    \begin{theorem} \label{thm::NZSGConv}
        Let $Z = (\agentset, \edgeset, (S_k)_{k \in \agentset}, (A^{kl}, A^{lk})_{(k,
    l) \in \edgeset})$ be a network zero sum game which, for some $T_1, \ldots, T_N > 0$, has unique QRE $\p \in \Delta$. Let $G = (\agentset, (S_k,
    u_k)_{k \in \agentset})$ be a game such that $d(Z, G) < \delta$ for some $\delta > 0$. Then, if $\x(t)$ is a trajectory of
    mixed strategies generated by running (\ref{eqn::QLD}) on $G$, 
    \begin{equation*}
        \lim_{t \rightarrow \infty} D_{KL}(\p || \x(t)) \leq \frac{N \delta}{T_{min}}
    \end{equation*}
    where $T_{\min} = \min_k T_k$.
    
    \end{theorem}
    
    Theorem \ref{thm::NZSGConv} provides a method by which the behaviour of Q-Learning dynamics can be understood even if
    game is slightly perturbed away from a NZSG. It is important to note that the approximate behaviour is also governed by
    choice of exploration rate $T_k$; in particular, the region to which (\ref{eqn::QLD}) converges decreases as $T_{\min}$
    increases. This can be explained as follows: as the exploration rate increases, each agent places less importance on the
    rewards that each action produces when updating their mixed distribution. Therefore, perturbations away from the NZSG
    condition are not felt as strongly as they would be if the exploration rate were low.
    
    In proving Theorem \ref{thm::NZSGConv}, we begin with the following proposition.
    
    \begin{proposition} \label{prop::rewarddist}
        Let $\Gamma_1 = (\agentset, (S_k, A_k)_{k \in \agentset})$ and $\Gamma_2 = (\agentset, (S_k, B_k)_{k \in
        \agentset})$ and let
        $\delta = d(\Gamma_1, \Gamma_2)$. Then, for any $k \in \agentset$ and $\x_k, \y_k \in \Delta_k$, $\x_{-k} \in \Delta_{-k}$ and any $T_k \geq 0$
        \begin{equation*}
            (\x_k - \y_k)^\top \left[a_k(\x_{-k}) - T_k \ln \x_k \right] \leq \delta + (\x_k - \y_k)^\top \left[b_k(\x_{-k}) - T_k \ln \x_k \right]
        \end{equation*}
        where $a_{ki}(\x) = \frac{\partial A_k (\x)}{\partial x_{ki}}$, $b_{ki}(\x) = \frac{\partial B_k (\x)}{\partial x_
        {ki}}$ define the rewards in $\Gamma_1$ and $\Gamma_2$ respectively. 
    \end{proposition}
    
    \begin{proof}
    Starting from the definition of MPD, we know that, for any $k \in \agentset$ and any $\x_k, \y_k \in \Delta_k$, $\x_{-k} \in \Delta_{-k}$
    \begin{align*}
            |A_k(\x_k, \x_{-k}) - A_k(\y_k, \x_{-k}) - (B_k(\x_k, \x_{-k}) - B_k(\y_k, \x_{-k}))| &\leq \delta \\
            |\x_k^\top a_k(\x_{-k}) - \y_k^\top a_k(\x_{-k}) - (\x_k^\top b_k(\x_{-k}) - \y_k^\top b_k(\x_{-k}))| & \leq \delta \\
            (\x_k - \y_k)^\top a_k(\x_{-k}) - (\x_k - \y_k)^\top b_k(\x_{-k}) &\leq \delta \\
            (\x_k - \y_k)^\top a_k(\x_{-k}) &\leq \delta + (\x_k - \y_k)^\top b_k(\x_{-k}) \\
            % (\x_k - \y_k)^\top a_k(\x_{-k}) - T_k  (\x_k - \y_k)^\top \ln \x_k &\leq \delta + (\x_k - \y_k)^\top b_k(\x_{-k})- T_k  (\x_k - \y_k)^\top \ln \x_k \\
            (\x_k - \y_k)^\top \left[a_k(\x_{-k}) - T_k \ln \x_k \right] &\leq \delta + (\x_k - \y_k)^\top \left[b_k(\x_{-k}) - T_k \ln \x_k \right]
    \end{align*}
    \end{proof}
    
    The following Lemma adapts the proof technique of \cite{piliouras:zerosum} in which it was shown that Q-Learning converges to a unique
    QRE in any NZSG. We extend this to consider games which do not necessarily fall into this rather restrictive class.
    \begin{lemma} \label{lem::KLNZSG}
        Let $Z$ and $G$ be games in the setting of Theorem \ref{thm::NZSGConv}. Then, if agents playing in game $G$ update
        their mixed
        strategies according to (\ref{eqn::QLD}), $D_{KL}(\p || \x(t))$ is strictly decreasing whenever $\x(t)$ satisfies
        \begin{equation} \label{eqn::KLdecreasecond}
            \frac{N \delta}{T_{\min}} < D_{KL}(\p || \x(t)) + D_{KL}(\x(t) || \p)
        \end{equation}
    \end{lemma}
    
    \begin{proof}
        From Lemma \ref{lem::QLRD}, we can write (\ref{eqn::QLD}) as the replicator system in the perturbed game $G^H$.
        Using this we can take the time derivative of $D_{KL}(\p_k || \x_k(t))$ giving
        \begin{align}
            \frac{d}{dt}D_{KL}(\p_k || \x_k(t)) &= \frac{d}{dt} \sum_{i} p_{ki} \ln \frac{p_{ki}}{x_{ki}(t)} \nonumber\\
            &= - \frac{d}{dt} \sum_{i} p_{ki} \ln x_{ki}(t) \nonumber \\
            &= - \sum_i p_{ki} \frac{\dot{x}_{ki}(t)}{x_{ki}(t)} \nonumber \\
            &= - \sum_i p_{ki} \left(r_{ki}^H(\x) - \langle \x_{k}, r_k^H(\x) \rangle \right) \nonumber \\
            &= (\x_k - \p_k)^\top \left[r_k(\x_{-k}) - T_k \ln \x_k \right]
        \end{align}
        in which $r_k$ denotes the rewards in the game $G$. From Proposition \ref{prop::rewarddist} it holds that
        \begin{align}
            \frac{d}{dt}D_{KL}(\p_k || \x_k(t)) &\leq \delta + (\x_k - \p_k)^\top \left[z_k(\x_{-k}) - T_k \ln \x_k \right] \nonumber
        \end{align}
        in which $\z_k$ denotes the rewards for the NZSG $Z$. We continue now along the lines of \cite{piliouras:zerosum}
        \begin{align}
            \frac{d}{dt}D_{KL}(\p_k || \x_k(t)) &\leq \delta + (\x_k - \p_k)^\top [z_k(\x_{-k}) - T_k \ln \x_k)] \nonumber\\
            &= \delta + (\x_k - \p_k)^\top [z_k(\x_{-k}) - z_k(\p_{-k}))] - T_k (\x_k - \p_k)^\top [\ln \x_k - \ln \p_k]
            \label{eqn::step1}\\
            &= \delta + (\x_k - \p_k)^\top [z_k(\x_{-k}) - z_k(\p_{-k}))] - T_k [D_{KL}(\p_k || \x_k(t)) + D_{KL}(\x_k(t) ||
            \p_k)] \label{eqn::step2}
        \end{align}
        where (\ref{eqn::step1}) follows due to \cite{piliouras:zerosum}, Theorem 3.2 and (\ref{eqn::step2}) is due to 
        \cite{piliouras:zerosum}, Property 1. Taking the sum over all $k \in \agentset$ yields
        \begin{equation*}
            \frac{d}{dt}D_{KL}(\p || \x(t)) \leq N\delta + \sum_k (\x_k - \p_k)^\top [z_k(\x_{-k}) - z_k(\p_{-k}))] -
            \sum_k T_k [D_{KL}(\p_k || \x_k(t)) + D_{KL}(\x_k(t) || \p_k)]
        \end{equation*}
        Now, under (\ref{eqn::KLdecreasecond})
        \begin{align*}
            &\frac{N \delta}{T_{\min}} < \sum_k [D_{KL}(\p_k || \x_k(t)) + D_{KL}(\x_k(t) || \p_k)] \\
            \implies& N\delta < T_{\min} \sum_k [D_{KL}(\p_k || \x_k(t)) + D_{KL}(\x_k(t) || \p_k)] \\
            \implies& N\delta - \sum_k T_k [D_{KL}(\p_k || \x_k(t)) + D_{KL}(\x_k(t) || \p_k)] < 0  
        \end{align*}
        In addition, from Lemma 4.3 of \cite{piliouras:zerosum}, it holds that
        \begin{equation*}
            \sum_k (\x_k - \p_k)^\top [z_k(\x_{-k}) - z_k(\p_{-k}))] = 0
        \end{equation*}
        Putting all this together, if (\ref{eqn::KLdecreasecond}) holds, then $\frac{d}{dt}D_{KL}(\p || \x(t))  < 0$
    \end{proof}
    
    \begin{proof}[Proof of Theorem \ref{thm::NZSGConv}]
         Define
        \begin{equation*}
            S = \left\{ \x \in \Delta \, | \, D_{KL}(\p || \x) + D_{KL}(\x || \p) \leq  \frac{N \delta}{T_{\min}}\right\}
        \end{equation*}
        \sloppy in which $T_{\min} = \min_k T_k$.
    
        For any $\x(t) \notin S$, it holds from Lemma \ref{lem::KLNZSG} that $D_{KL}(\p || \x(t))$ is strictly decreasing.
        By definition it is also bounded below by zero. It holds, then, that $\x(t)$ reaches $S$ in finite time. At this
        time, $D_{KL}(\p || \x(t)) \leq \sup_{\x \in S} D_{KL}(\p || \x) =: D_S$. Furthermore, by Lemma \ref{lem::KLNZSG},
        if $\x(t)$ leaves $S$, $D_{KL}(\p || \x(t))$ cannot increase past $D_S$. It follows, then, that $\limsup_{t
        \rightarrow \infty} D_{KL}(\p || \x(t)) \leq D_S$. Finally, we note that $D_S = \sup_{\x \in S} D_{KL}(\p || \x)
        \leq \sup_{\x \in S} D_{KL}(\p || \x) + D_{KL}(\x || \p) \leq \frac{N \delta}{T_{\min}}$.
    \end{proof}
    
    \begin{remark} It is important to note that Theorem \ref{thm::NZSGConv} makes no statement on whether Q-Learning in a near NZSG will itself converge
 to a QRE. In fact such counter-examples are demonstrated in Figure \ref{fig::Conflict_Network}, and complex behaviour is known
 to be prevalent in multi-agent learning (e.g. \cite{sanders:chaos,cheung:decomposition}). Nonetheless, Theorem 
    \ref{thm::NZSGConv} provides a complete picture on the \emph{approximate} last iterate behaviour of Q-Learning. It does this by determining a region to which Q-Learning dynamics must remain trapped, even if it does not ultimately reach a QRE within this region. This region is defined with respect to the QRE of an NZSG, which is unique and can be found by running Q-Learning.
\end{remark}
    
     % subsection approximate_convergence (end)
    
    \subsection{Finding the Closest NZSG} % (fold)
    \label{sub:finding_the_closest_nzsg}
    
    In order use Theorem \ref{thm::NZSGConv} to determine the approximate behaviour of an arbitrary competitive (but not zero sum) game, it is first
    required that we find the nearest network zero-sum game. In this section we show that this process can be solved efficiently. In particular, given
    any network game $\Gamma = (\agentset, \edgeset, (S_k)_{k \in \agentset}, (A^{kl}, A^{lk})_{(k, l) \in \edgeset})$ which is not necessarily zero sum,
    the problem of finding the `nearest' NZSG can be written as a quadratic minimisation problem with linear constraints. In doing so, the approximate
    behaviour of Q-Learning in the original game can be determined. 
    
    This formulation manipulates the result of \cite{cai:minimax}: that any NZSG is payoff equivalent to a pairwise constant-sum game, where the constants
    add to zero. More formally, this can be stated as the following proposition.
    
    \begin{proposition}[\cite{piliouras:zerosum}, \cite{cai:minimax}]
        Let $Z = (\agentset, \edgeset, (S_k)_{k \in \agentset}, (A^{kl}, A^{lk})_{(k,l) \in \edgeset})$ be a NZSG. For all $(k, l) \in \edgeset$
        there exist $(\hat{A}^{kl}, \hat{A}^{lk})$ and a constant $c_{kl} \in \R$ such that
        \begin{equation}
            [\hat{A}^{kl}]_{ij} + [\hat{A}^{lk}]_{ji} = c_{kl}, \; \forall i \in S_k, j \in S_l,
        \end{equation}
        with
        \begin{equation}
            \sum_{(k, l) \in \edgeset} c_{kl} = 0,
        \end{equation}
        and payoffs to agent $k$ in $Z$ is equivalent to their payoffs in $\hat{Z} = (\agentset, \edgeset, (S_k)_{k \in \agentset}, (\hat{A}^{kl}, \hat{A}^{lk})_{(k,l) \in \edgeset})$. In particular, for all $k \in \agentset$ and all $\x_k \in \Delta_k$
        \begin{equation}
            \sum_{(k, l) \in \edgeset} \x_k^\top \hat{A}^{kl} \x_l = \sum_{(k, l) \in \edgeset} \x_k^\top A^{kl} \x_l
        \end{equation}
    \end{proposition}
    
    As such, given the network game $\Gamma$, we can write the problem of finding the `nearest' NZSG as finding the nearest pairwise constant-sum game. This is formulated as
    
    \begin{equation} \tag{P1} \label{eqn::P1}
        \begin{cases}
            \min_{(\hat{A}^{kl}, \hat{A}^{lk}, c_{kl})_{(k, l) \in \edgeset}} &\sum_{(k, l) \in \edgeset} ||\hat{A}^{kl} - A^{kl}||_2^2 + ||\hat{A}^{lk} - A^{lk}||_2^2 \\
            \text{s.t. } & [A^{kl}]_{ij} + [A^{lk}]_{ji} = c_{kl}, \; \forall (k, l) \in \edgeset, \, \forall i \in S_k, \, j \in S_l \\
            & \sum_{(k, l) \in \edgeset} c_{kl} = 0
        \end{cases}
    \end{equation}
    
    where $A^{kl}, A^{lk}$ are the payoff matrices which define $\Gamma$. As the objective function in (\ref{eqn::P1}) is quadratic, and the constraints are
    linear, (\ref{eqn::P1}) is a quadratic optimisation problem which can be solved efficiently. 
    
    To connect the minimisation of the $2-$norm to (\ref{eqn::MPD}), we 
    have the following results.
    \begin{proposition} \label{prop::reward-MPD}
        Suppose $\Gamma_1 = (\agentset, (S_k, A_k)_{k \in \agentset})$, $\Gamma_2 = (\agentset, (S_k, B_k)_{k \in \agentset})$
        are games which have rewards $a_{ki}(\x_{-k}) = \partial A_{ki}(\x)/\partial x_{ki}$ and $b_{ki}(\x_{-k}) = \partial B_{ki}(\x)/\partial x_{ki}$ respectively. Suppose also that, for all $k \in \agentset$, $i \in S_k$ and $\x_{-k} \in \Delta_{-k}$,
        \begin{equation} \label{eqn::reward-MPD}
            \left| a_{ki}(\x_{-k}) - b_{ki}(\x_{-k}) \right| \leq \frac{\delta}{2 n_k}
        \end{equation}
        where $\delta > 0$. Then $d(\Gamma_1, \Gamma_2) \leq \delta$
    \end{proposition}
    
    \begin{proof}
        For any agent $k$, any $\x_k, \y_k \in \Delta_k$ and any $\x_{-k} \in \Delta_{-k}$ it holds that
        \begin{align*}
            &|A_k(\y_k, \x_{-k}) - A_k(\x_k, \x_{-k}) - (B_k(\y_k, \x_{-k}) - B_k(\x_k, \x_{-k}))| \\
            =& |\y_k^\top a_k(\x_{-k}) -\x_k^\top a_k(\x_{-k}) - (\y_k^\top b_k(\x_{-k}) -\x_k^\top b_k(\x_{-k}))| \\
            \leq& |\y_k^\top (a_k(\x_{-k}) - b_k(\x_{-k})) -\x_k^\top (a_k(\x_{-k}) - b_k(\x_{-k}))| \\
            =& \left| \sum_i y_{ki} (a_{ki}(\x_{-k}) - b_{ki}(\x_{-k})) \right| + \left| \sum_i x_{ki} (a_{ki}(\x_{-k}) - b_{ki}(\x_{-k})) \right| \\
            \leq& \sum_i |y_{ki} (a_{ki}(\x_{-k}) - b_{ki}(\x_{-k}))| + \sum_i |x_{ki} (a_{ki}(\x_{-k}) - b_{ki}(\x_{-k}))| \\
            \leq& \sum_i |a_{ki}(\x_{-k}) - b_{ki}(\x_{-k})| + \sum_i |a_{ki}(\x_{-k}) - b_{ki}(\x_{-k})| \leq \delta
        \end{align*}
        where the final inequality holds due to (\ref{eqn::reward-MPD}).
    \end{proof}
    
    From Proposition~\ref{prop::reward-MPD} we immediately obtain the following corollary for the particular case of network games.
    
    \begin{corollary} \label{corr::abs-distance}
        Suppose that, in the setting of Proposition \ref{prop::rewarddist}, $\Gamma_1$ and $\Gamma_2$ are network games whose rewards are defined through
        the payoff matrices $(A^{kl}, A^{lk})_{(k, l) \in \edgeset}$, $(B^{kl}, B^{lk})_{(k, l) \in \edgeset}$ respectively. Suppose also that, for all $
        (k, l) \in \edgeset$, $i \in S_k$ and $j \in S_l$
        \begin{equation}
            \left| (A^{kl})_{ij} - (B^{kl})_{ij} \right| \leq \frac{\delta}{2 n_k \sum_{(k, l) \in \edgeset} n_l}
        \end{equation}
        where $\delta > 0$. Then $d(\Gamma_1, \Gamma_2) \leq \delta$
    \end{corollary}
    
    \begin{proof}
        \begin{align*}
            &|a_{ki}(\x_{-k}) - b_{ki}(\x_{-k})| \\
            \leq& |\sum_{(k, l) \in \edgeset} \sum_{j \in S_l} (A^{kl} - B^{kl})_{ij} x_{lj}| \\
            \leq& \sum_{(k, l) \in \edgeset} \sum_{j \in S_l} |(A^{kl} - B^{kl})_{ij}| \leq \frac{\delta}{2 n_k}
        \end{align*}
        Then the result follows from Proposition \ref{prop::rewarddist}.
    \end{proof}
    
    \begin{corollary} \label{corr::2-norm}
        Suppose that, in the setting of Proposition \ref{corr::abs-distance}, $\Gamma_1$, $\Gamma_2$ are such that for all $(k, l) \in \edgeset$
        \begin{equation}
            ||A^{kl} - B^{kl}||_2 \leq \frac{\delta}{2 n_k \sum_{(k, l) \in \edgeset} n_l}
        \end{equation}
        where the matrix norm for a matrix $A \in M_{m \times n (\R)}$ is given by $||A||_2 = \sup_{\x \in \R^n \, :\, ||x||_2 \neq 0} \frac{||A \x||_2}
        {||\x||_2}$. Then $d(\Gamma_1, \Gamma_2) \leq \delta$.
    \end{corollary}
    
    \begin{proof}
        An equivalent definition of the $2$-norm is given by
        \begin{equation*}
            ||A||_2 = \sup \{ \x^\top A \y \, : \, \x, \y \in \R^n, ||\x||_2 = 1, ||\y||_2 = 1 \}
        \end{equation*}
        Now, for any $k \in \agentset$ and $i \in S_k$ let $\e_i$ be the $i$'th unit vector in $R^{n_k}$, i.e. $\e_i$ has all zeros except for in the
        $i$'th entry, where it is one. Clearly, $||\e_i||_2 = 1$, so that, for any $(k, l) \in \edgeset$, $i \in S_k$, $j \in S_l$
        \begin{equation*}
            ||A^{kl} - B^{kl}||_2 \geq \e_i^\top (A^{kl} - B^{kl}) \e_j = (A^{kl})_{ij} - (B^{kl})_{ij}
        \end{equation*}
        Applying also that $||A^{kl} - B^{kl}||_2 = ||B^{kl} - A^{kl}||_2$, if $||A^{kl} - B^{kl}||_2 \leq \frac{\delta}{2 n_k \sum_{(k, l) \in \edgeset}
        n_l}$, it holds that $\left| A^{kl}_{ij} - B^{kl}_{ij} \right| \leq \frac{\delta}{2 n_k \sum_{(k, l) \in \edgeset} n_l}$. The result then follows
        from Corollary \ref{corr::abs-distance}.
    \end{proof}
    
    Using the process outlined in this section, it is possible to determine approximate convergence in competitive, but not zero sum, network games. Its advantage lies in the fact that
    the QRE of NZSGs are unique for any $T_k > 0$ and it is known that Q-Learning, for any initial condition must converge to this QRE 
    \cite{piliouras:zerosum}. Therefore, the aforementioned process provides a method to determine approximate convergence of Q-Learning in $\Gamma$ for
    any initial condition.
    
    % section near_network_zero_sum_games (end)
    
    \section{Experiments on Near NZSG} % (fold)
    \label{sub:experiments_on_near_nzsg}
    
    In our experiments we examine the implications of Theorem 1. In particular we confirm that Q-Learning in near NZSG asymptotically remain close to the QRE of the NZSG. We also examine the implication of this finding for the introduction of noise in the payoffs. 
    
    \paragraph{Visualising Theorem \ref{thm::NZSGConv}} % (fold)
    \label{par:visualising_theorem_thm::nzsgconv}
    
    \begin{figure*}[t!]
    \centering
         \begin{subfigure}[b]{0.225\textwidth}
             \centering
             \includegraphics[width=\textwidth]{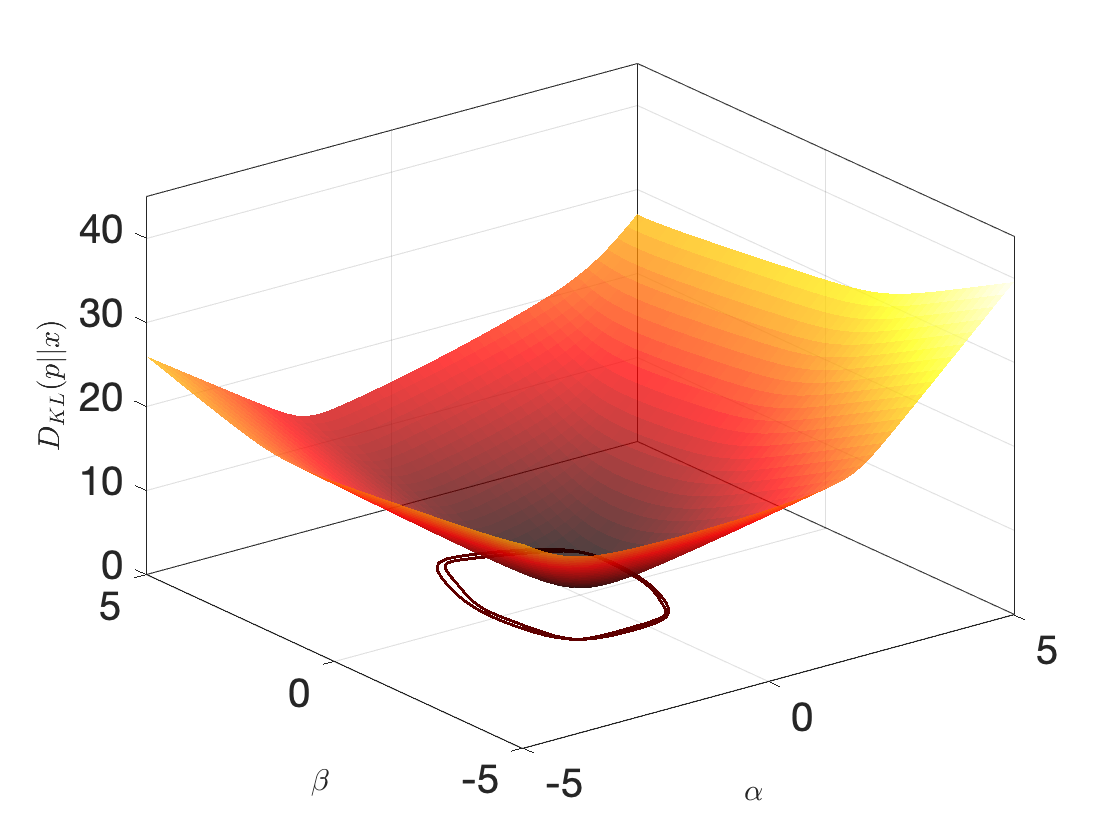}
             \caption*{N = 5}
         \end{subfigure}
         \begin{subfigure}[b]{0.225\textwidth}
             \centering
             \includegraphics[width=\textwidth]{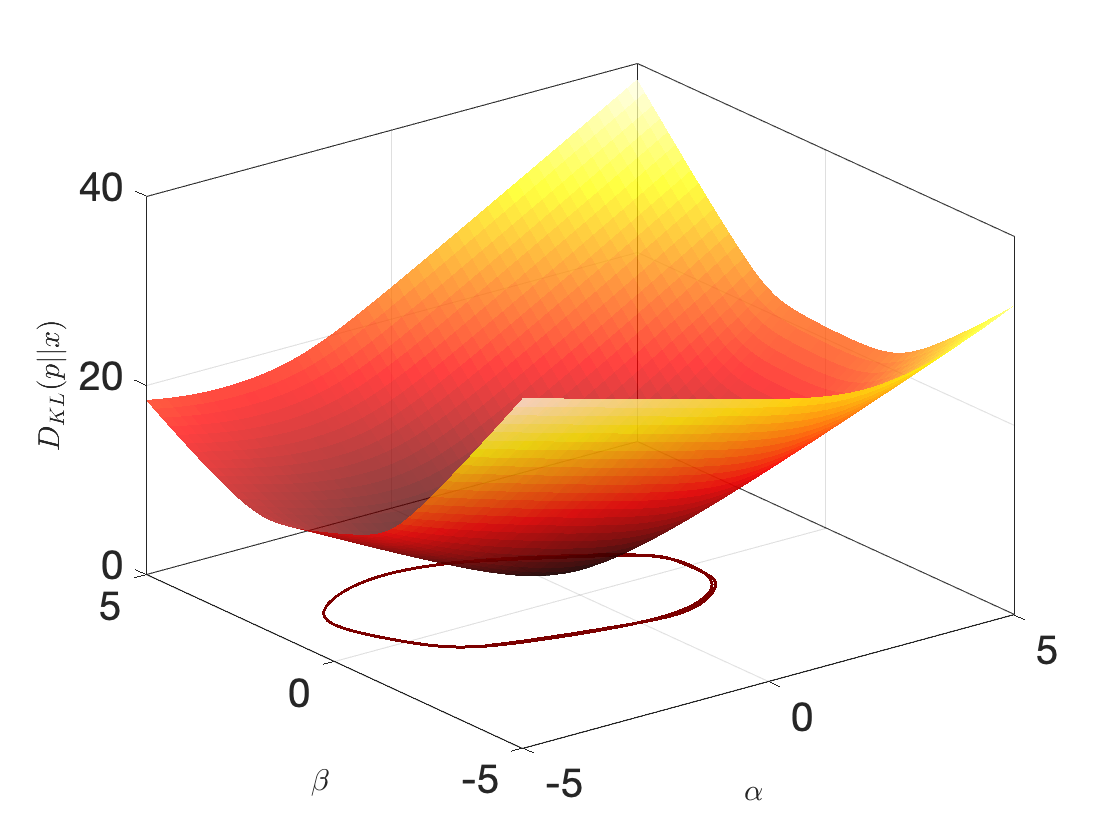}
             \caption*{N = 7}
         \end{subfigure}
         \begin{subfigure}[b]{0.225\textwidth}
             \centering
             \includegraphics[width=\textwidth]{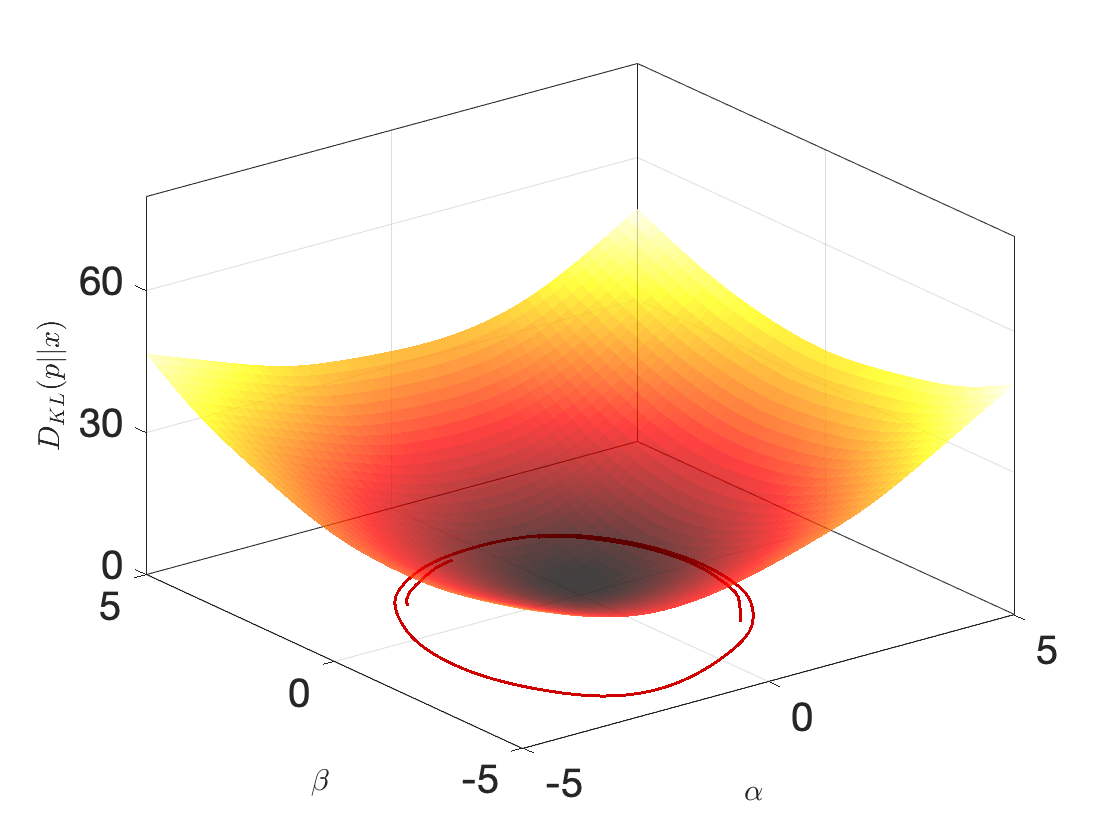}
             \caption*{N = 13}
         \end{subfigure}
         \begin{subfigure}[b]{0.225\textwidth}
             \centering
             \includegraphics[width=\textwidth]{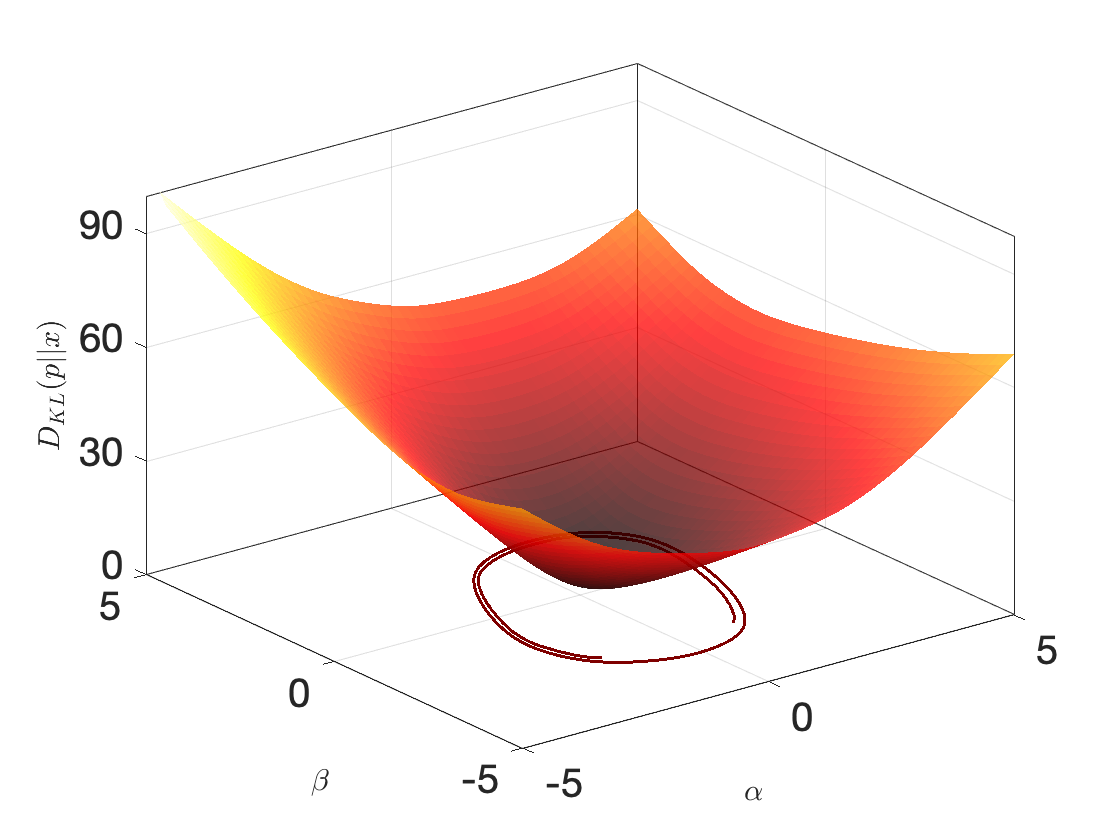}
             \caption*{N = 17}
         \end{subfigure}
            \caption{Visualisation of the KL-divergence between the unique QRE and a mixed strategy in an NZSG, alongside a depiction of the region
            to which Q-Learning converges in nearby games. The minimum of the KL-divergence occurs at zero and the region is a neighbourhood around
            the minimiser (i.e. the QRE of the NZSG). In all cases, we choose $\delta = 1$ and $T = 0.75$ whilst we vary the number of players $N$.}
            \label{fig::KLplots}
    \end{figure*}

    In Figure \ref{fig::KLplots} we visualise the region to which Q-Learning converges as predicted by Theorem \ref{thm::NZSGConv}. In particular, we generate a two-action network zero-sum game for a given number of agents. We then plot the KL-divergence from the QRE $\p$ for a given exploration
    rate, using the dimensionality reduction technique of \cite{li:visualise}, which was adapted for the KL-Divergence by \cite{piliouras:zerosum}. The
    procedure is as follows. We first run the Q-Learning dynamics to determine the QRE in the NZSG. Then, we generate two random vectors
    $u, v \in [0, 1]^N$ which denote the probability with each agent plays their first action. We transform both $u, v$ as
    \begin{align*}
        \Tilde{u}_k = \ln \frac{u_k}{1 - u_k}, \; \; \Tilde{v}_k = \ln \frac{v_k}{1 - v_k}
    \end{align*}
    Next, we take a linear combination of these transformed vectors to yield $z = \alpha \Tilde{u} + \beta \Tilde{v}$ for some choice of 
    $\alpha, \beta \in \R$. Finally, the point $z$ is mapped back into
    the unit simplex via
    \begin{align*}
        \Tilde{z}_k = \frac{\exp(z_k)}{1 + \exp{z_k}}
    \end{align*}
    This becomes the point against which we measure the KL Divergence to the QRE $\p$. The complete algorithm for this process can also be found in \cite{piliouras:zerosum}.
    We plot, on the $x-y$ plane, the contour $D_{KL}(\p || \z) = \frac{N \delta}{T_{\min}}$ for some choice
    of $\delta, T_{\min}$. It is clear that this forms a neighbourhood around the QRE of the NZSG; the implication of Theorem \ref{thm::NZSGConv} is that, in games which are at most $\delta$ away from the NZSG, Q-Learning will asymptotically remain trapped in this neighbourhood.
    
    % paragraph visualising_theorem_thm::nzsgconv (end)
    
    \begin{figure}[tb]
        \centering
        \includegraphics[width=0.5 \textwidth]{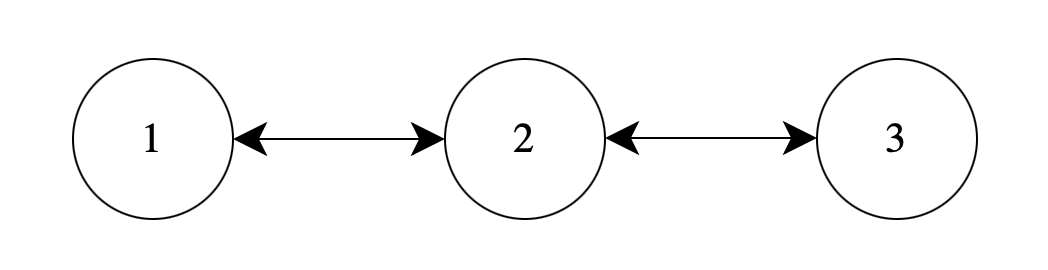}
        \caption{Three Player Chain Network Game}
        \label{fig::3pchain}
    \end{figure}
    
    \paragraph{Three Player Chain} % (fold)
    \label{par:three_player_chain}

    \begin{figure*}[t]
    \centering
         \begin{subfigure}[b]{0.45\textwidth}
             \centering
             \includegraphics[width=0.65\textwidth]{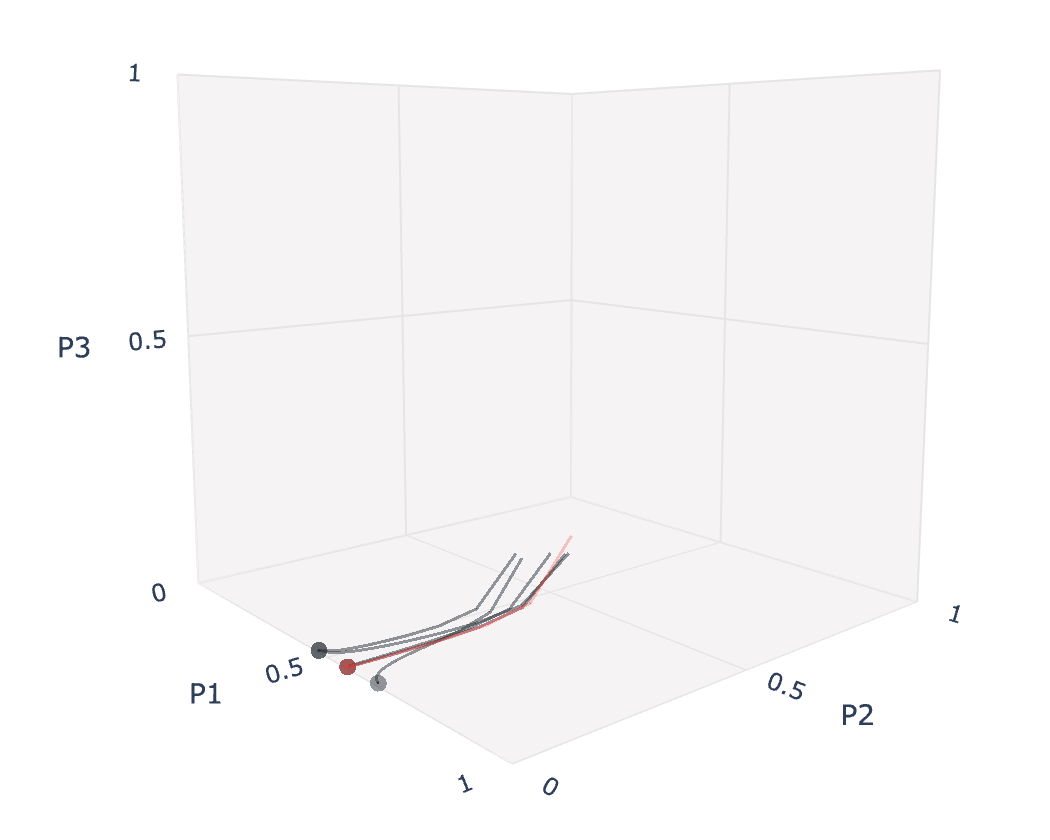}
              \caption*{$\delta \leq 0.75$}
         \end{subfigure}
         \begin{subfigure}[b]{0.45\textwidth}
             \centering
             \includegraphics[width=0.75\textwidth]{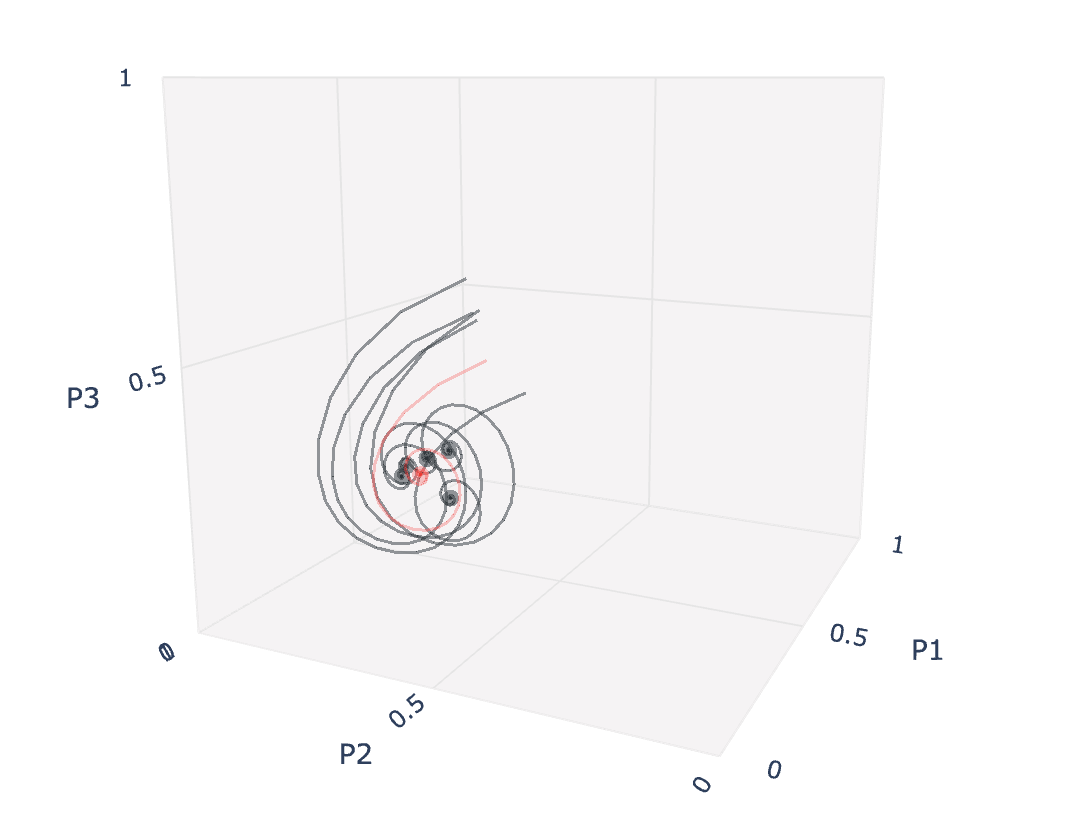}
              \caption*{$\delta \leq 2$}
         \end{subfigure}
    
            \caption{Trajectories of Q-Learning in near NZSG. In each plot, the red line depicts Q-Learning in an NZSG and the black 
        depicts Q-Learning in a nearby game which is not zero sum. Q-Learning converges to an equilibrium in
        the near-NZSG (black marker), where the equilibrium is `close' to the QRE of the NZSG (red marker). In all cases $T = 0.75$.}
            \label{fig::NearEq}
    \end{figure*}
    
    \begin{figure*}[t]
    \centering
        \begin{subfigure}[b]{0.3\textwidth}
             \centering
             \includegraphics[width=\textwidth]{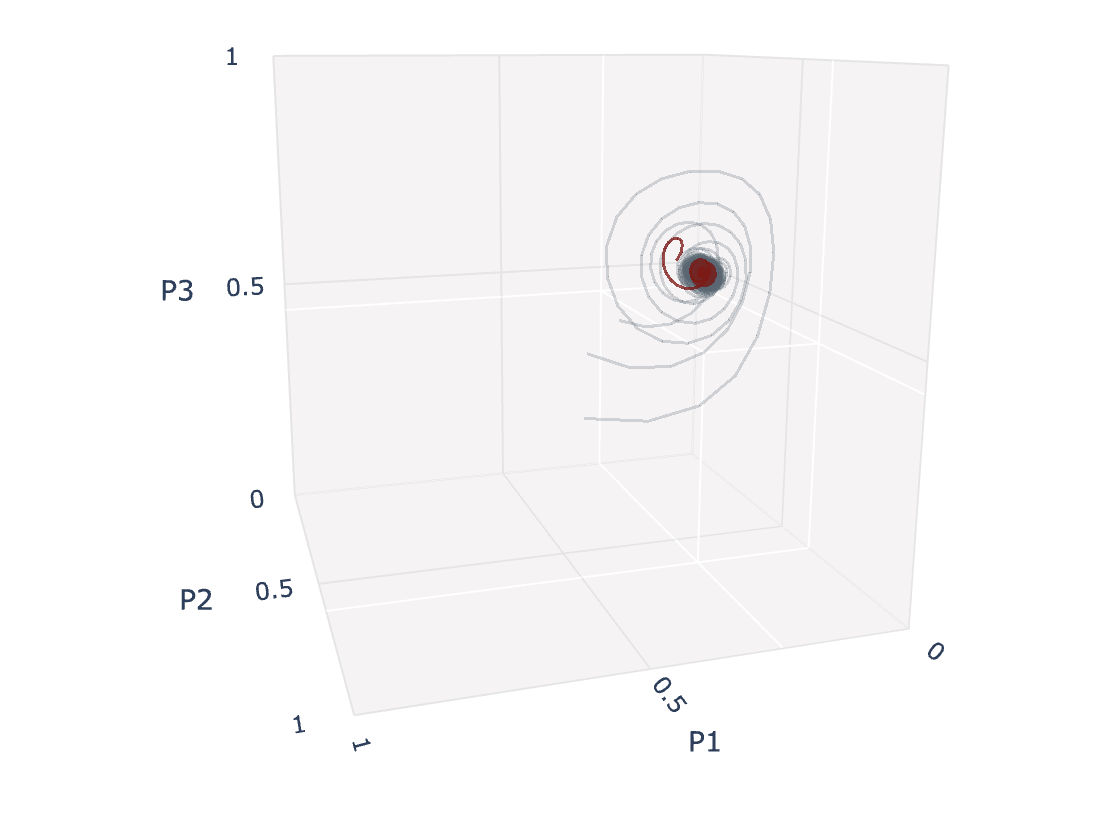}
             \caption*{$\delta \leq 075$}
         \end{subfigure}
         \begin{subfigure}[b]{0.3\textwidth}
             \centering
             \includegraphics[width=\textwidth]{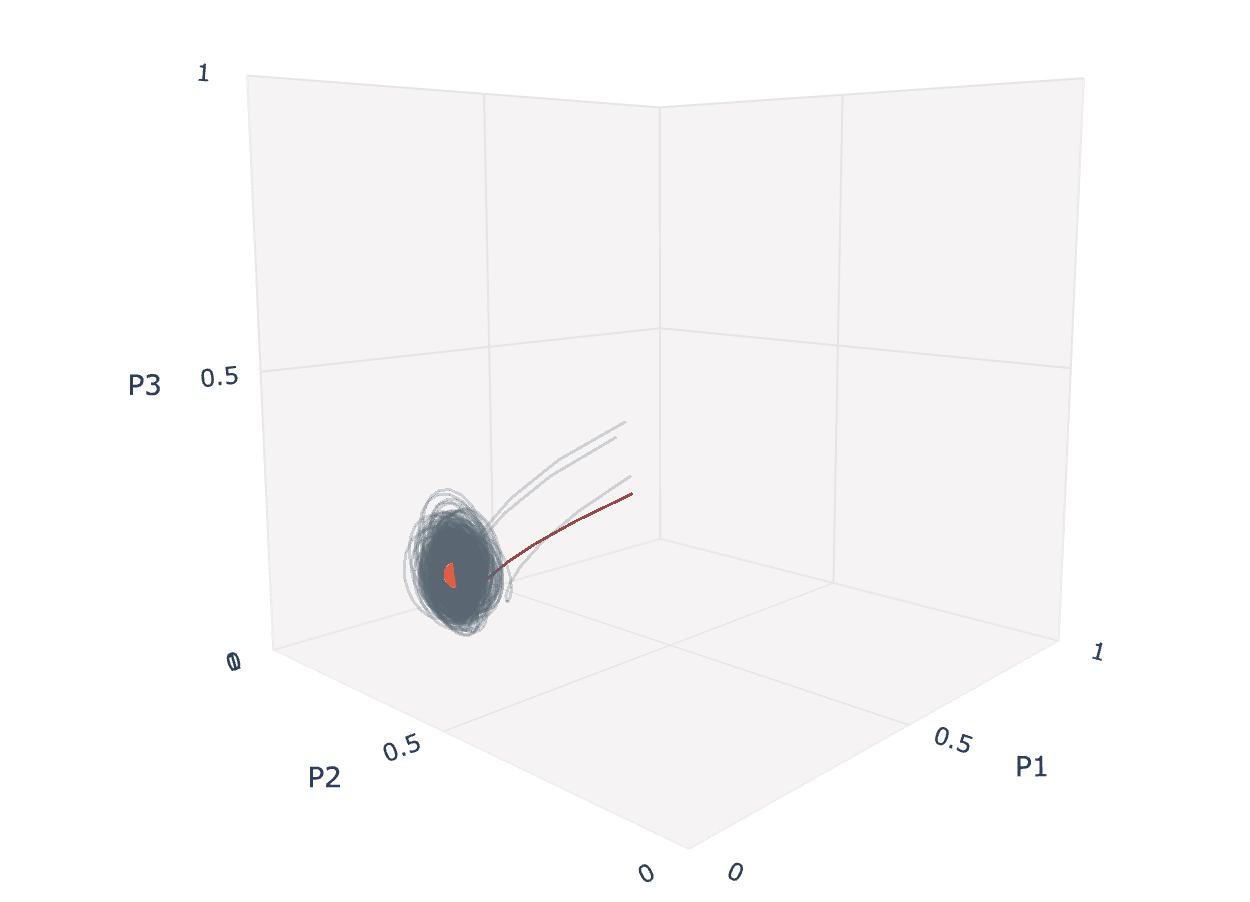}
             \caption*{$\delta \leq 2$}
         \end{subfigure}
         \begin{subfigure}[b]{0.3\textwidth}
             \centering
             \includegraphics[width=\textwidth]{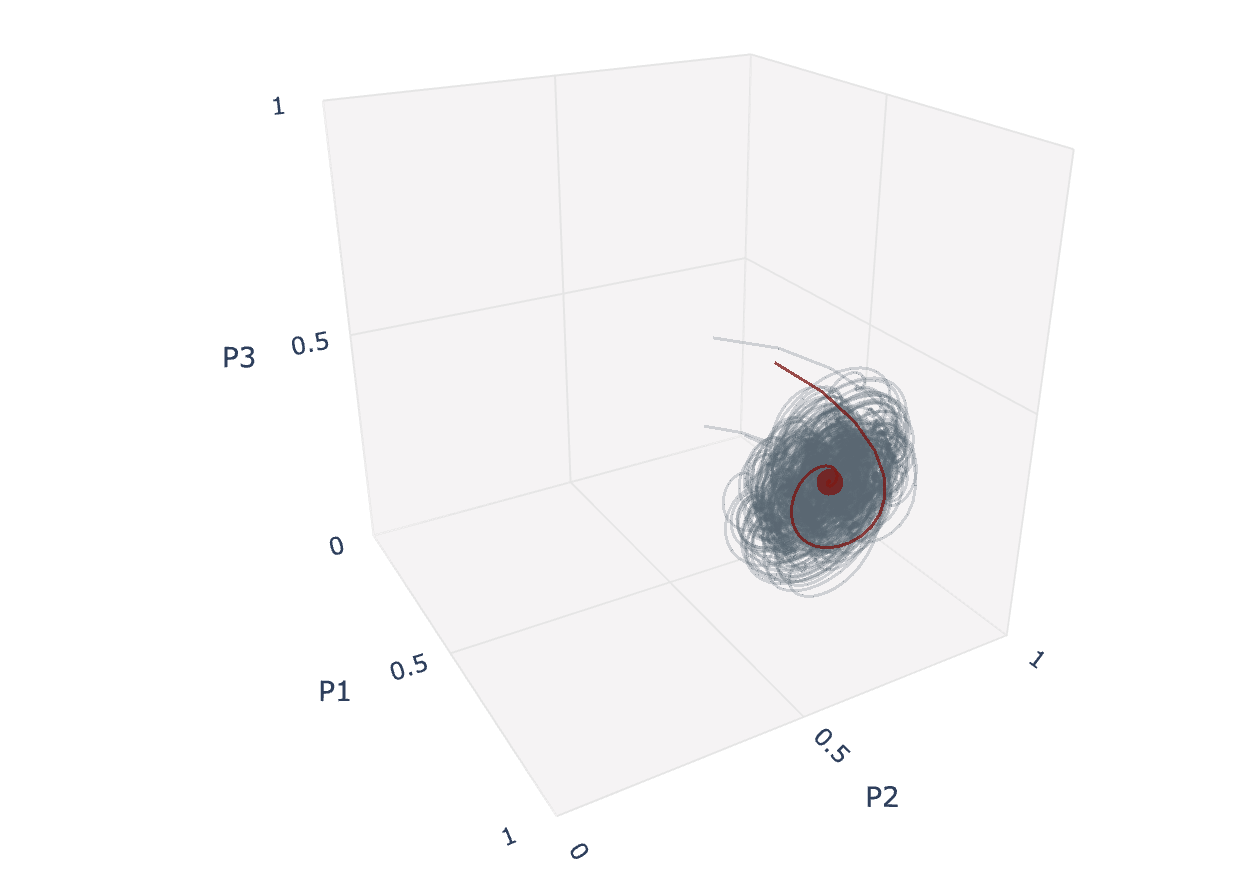}
             \caption*{$\delta \leq 3$}
         \end{subfigure}
            \caption{Trajectories of Q-Learning on an NZSG in the presence of additive noise. The noise is such that the perturbed game is always close to
            the NZSG. In this case, Q-Learning does not reach an fixed point, but will still remain asymptotically within a region surrounding the QRE
             of the NZSG (red marker). In all cases $T = 0.75$.}
            \label{fig::Noise}
    \end{figure*}
    
    We examine a `chain' network with three agents connected as in Figure \ref{fig::3pchain} where each agent has two actions. We generate a zero-sum game
    and run Q-Learning on this game to find its QRE \cite{piliouras:zerosum}. For the sake of simplicity we assume that all agents have the same
    exploration rate $T_k$ so that we replace the notation $T_k$ or $T_{\min}$ with simply $T$. Then, we perturb the payoff matrices to generate five
    near zero-sum games. We can use Corollary \ref{corr::abs-distance} to determine an upper bound on the distance between these games from the NZSG in
    terms of (\ref{eqn::MPD}).
    
    The results from this experiment are shown in Figure \ref{fig::NearEq}. The figures plot the probabiliy by which each player plays their first
    action. In all cases, the near-NZSG converge to equilibria (depicted with black markers) who are close to the QRE of the NZSG (red marker). The
    distance of the QRE of the pertubed games from the original increases as $\delta$ is increased from $0.75$ to $2$.
    
    When examining the effect of noise, we take the same network game setup and periodically (every 50 iterations) add noise to the payoff matrices to
    perturb the game away from the zero sum. By ensuring that the perturbations satisfy Corollary $\ref{corr::abs-distance}$ for some $\delta$, we can
    determine an upper bound on (\ref{eqn::MPD}). The results are shown in Figure \ref{fig::Noise}. The power of Theorem \ref{thm::NZSGConv} is
    apparent in this setting since, in this case, Q-Learning will not converge to an equilibrium. Despite this, since the perturbations are upper bounded
    by $\delta$, Theorem \ref{thm::NZSGConv} enforces that the trajectories remain within the neighbourhood of the QRE of the original game. This ensures
    the robustness of Q-Learning under the presence of noise. Note that, whilst in the experiments we use additive noise, Theorem \ref{thm::NZSGConv}
    makes no such assumption. The only requirement is that the perturbations are bounded. Of course, the larger this bound is, the larger the
    neighbourhood, as evidenced by the increase in spread of the Q-Learning trajectories as $\delta$ is increased.
    % paragraph three_player_chain (end)
    
    \paragraph{Ten Player Network} % (fold)
    \label{par:ten_player_network}
    
    \begin{figure*}[t!]
    \centering
         \begin{subfigure}[b]{0.45\textwidth}
             \centering
             \includegraphics[width=1.1\textwidth]{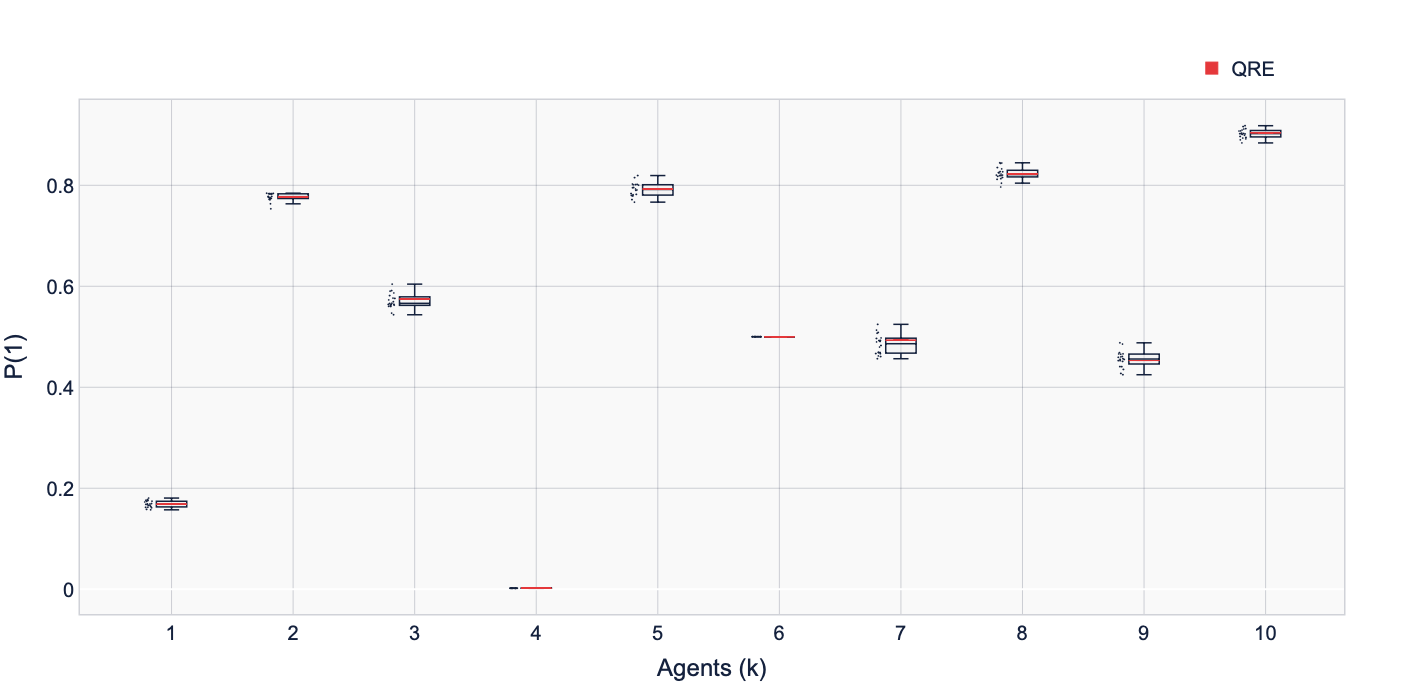}
             \caption*{$\delta \leq 2$}
         \end{subfigure}
         \begin{subfigure}[b]{0.45\textwidth}
             \centering
             \includegraphics[width=1.1\textwidth]{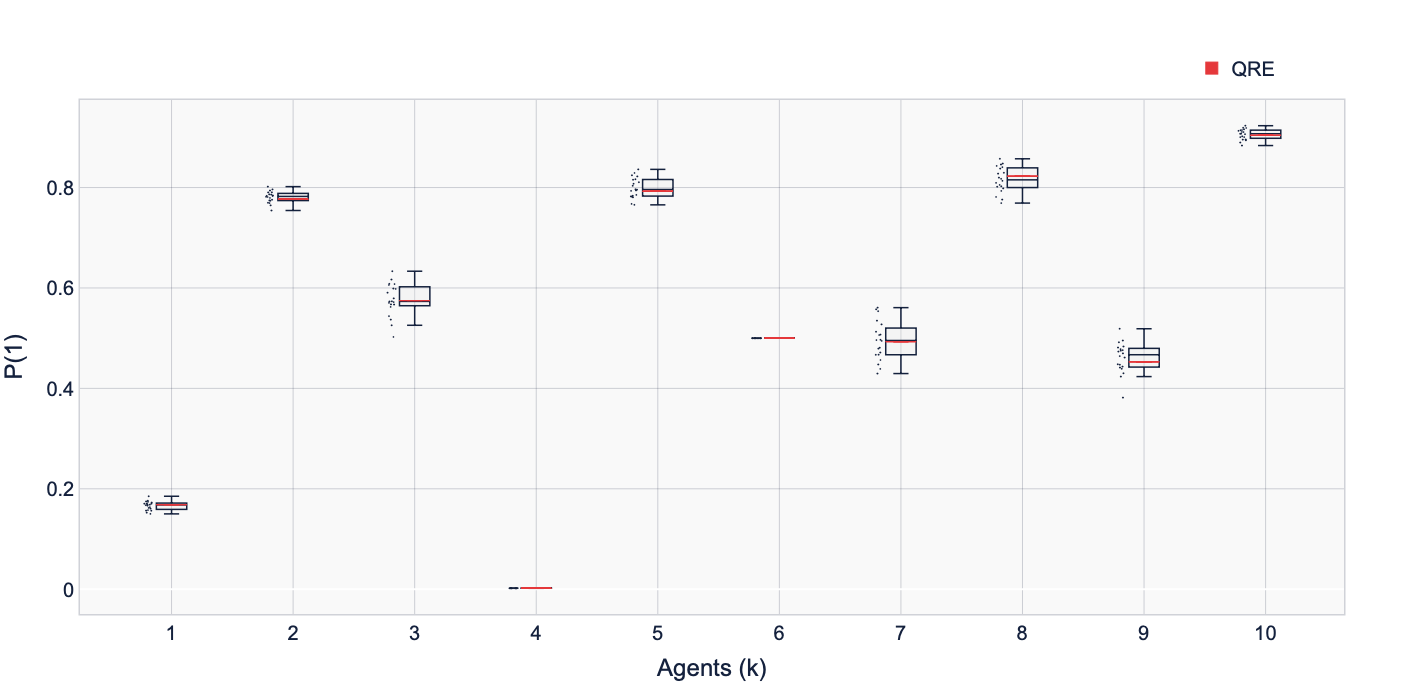}
             \caption*{$\delta \leq 3$}
         \end{subfigure}
         \begin{subfigure}[b]{0.45\textwidth}
             \centering
             \includegraphics[width=1.1\textwidth]{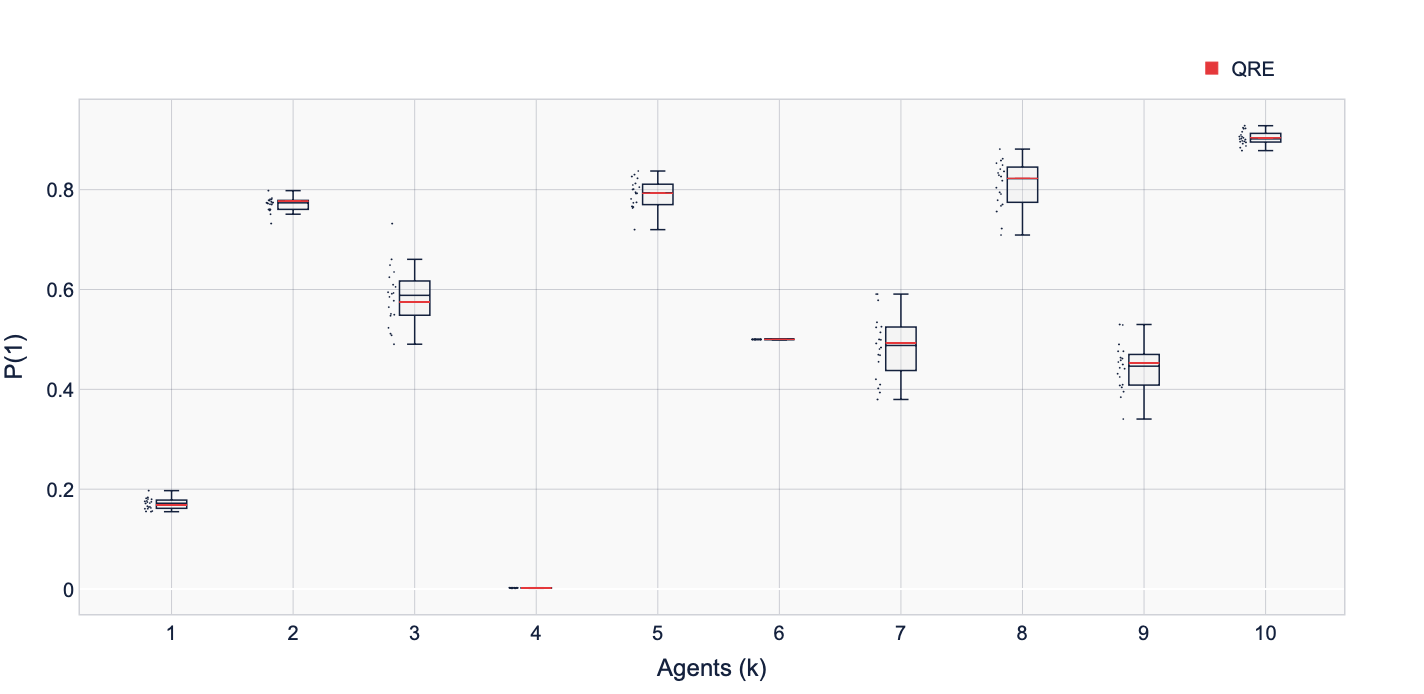}
             \caption*{$\delta \leq 5$}
         \end{subfigure}
         \begin{subfigure}[b]{0.45\textwidth}
             \centering
             \includegraphics[width=1.1\textwidth]{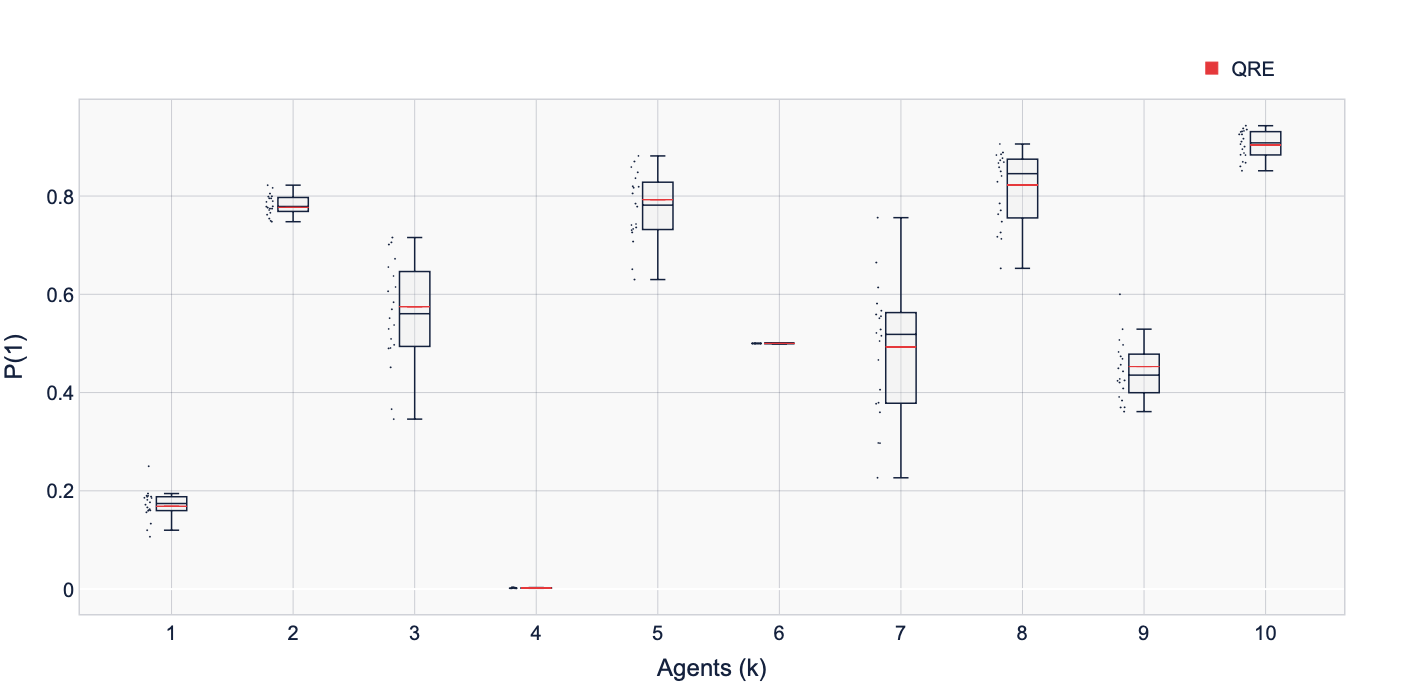}
             \caption*{$\delta \leq 9$}
         \end{subfigure}
            \caption{Summary Statistics of Last-Iterates of Q-Learning in 100 near-NZSGs. The y-axis depicts the probability which each agent assigns to
             their first action. The red line depicts the QRE of the NZSG whilst the box depicts the spread of last iterates in the near-NZSG. Markers
             next to the boxes depict the last iterate after $1 \times 10^6$ iterations in each game. In all cases $T=0.75$.}
            \label{fig::Barplots}
    \end{figure*}
    
    Finally, we extend our analysis to a 10-agent network where agents have two actions. In this case, the Q-Learning dynamics evolve in $\R^{20}$ and so
    it is not possible to visualise the trajectories. Rather, we generate a NZSG (the network and payoffs can be found in the Supplement) and 100
    near-NZSG which are generated in the same manner as the three-player chain network. Figure \ref{fig::Barplots} shows a summary of the last iterates
    of Q-Learning in 100 randomly generated near-zero sum game after $1 \times 10^6$ iterations. The behaviour agrees with the results in the lower
    dimensional case. In particular, it is clear that the last iterate of Q-Learning for all nearby games is within a bounded region around the QRE of
    the NZSG.
    % paragraph ten_player_network (end)
    
    \paragraph{Approximate Behaviour in a Conflict Network} % (fold)
    \label{par:approximate_behaviour_of_a_conflict_network}
    \begin{figure*}[b!]
    \centering
        \begin{subfigure}[b]{0.3\textwidth}
             \centering
             \includegraphics[width=\textwidth]{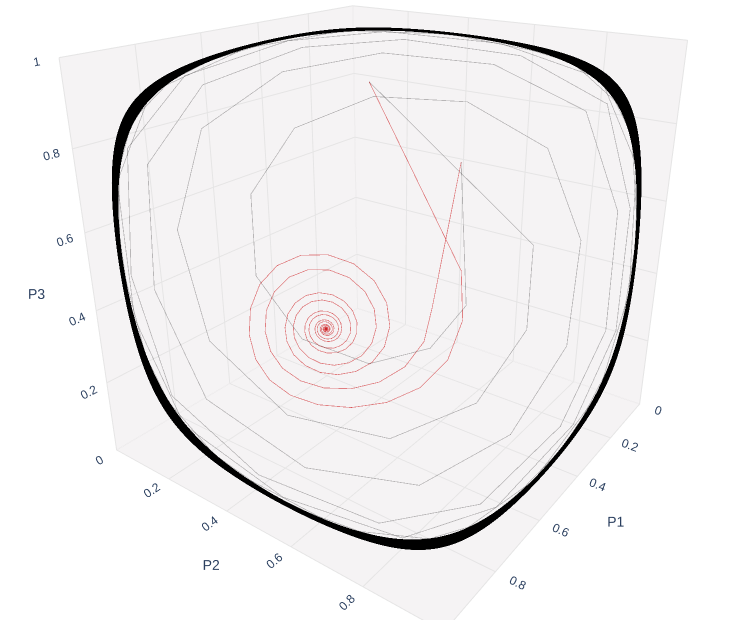}
             \caption*{$T = 0.35$}
         \end{subfigure}
         \begin{subfigure}[b]{0.3\textwidth}
             \centering
             \includegraphics[width=\textwidth]{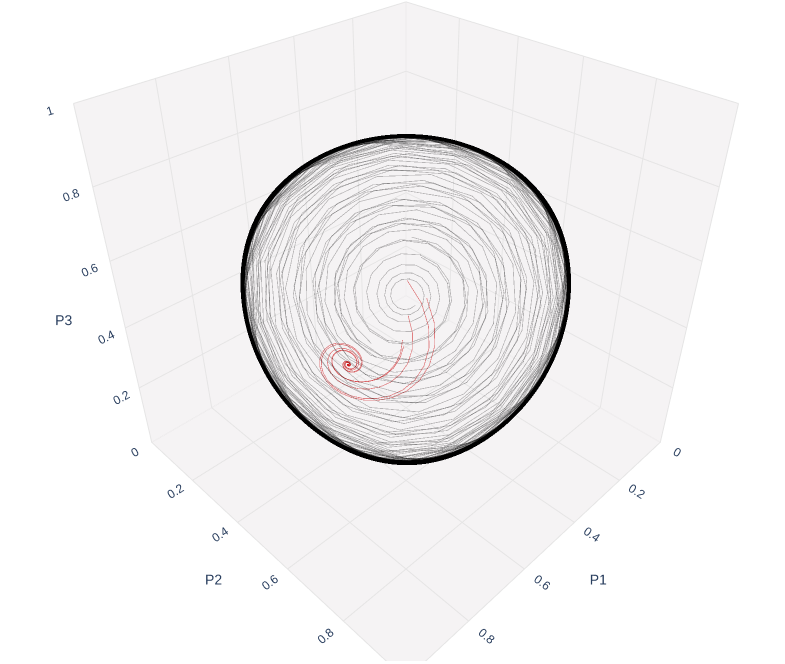}
             \caption*{$T = 0.5$}
         \end{subfigure}
         \begin{subfigure}[b]{0.3\textwidth}
             \centering
             \includegraphics[width=\textwidth]{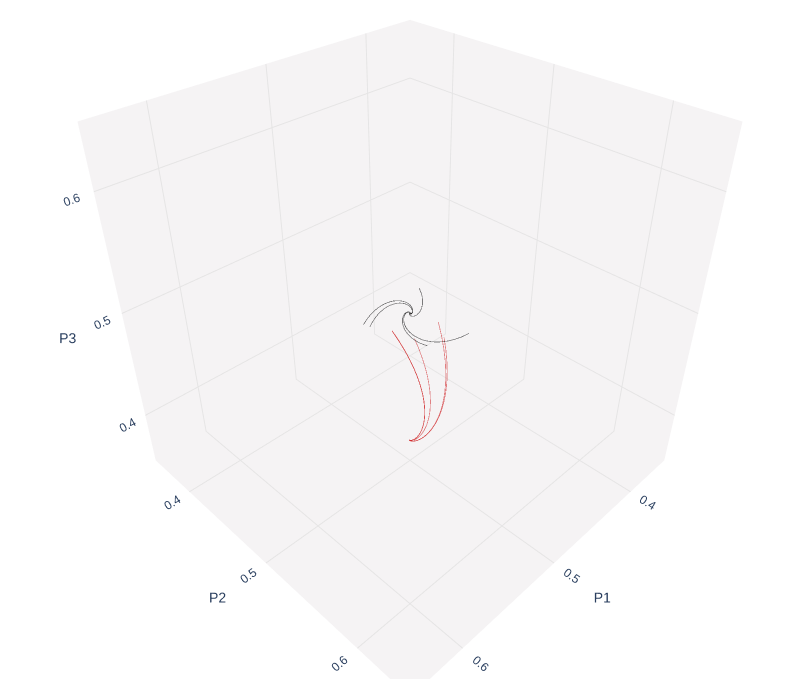}
             \caption*{$T = 2.5$}
         \end{subfigure}
            \caption{Dynamics of Q-Learning in a Conflict Network. Axes indicate the probability with which each agent
            plays their first action. Black trajectories denote the dynamics in the conflict network. Red trajectories denote
            the dynamics in the nearest network zero sum game. Q-Learning dynamics in the NZSG converge to a QRE. In the conflict network,
            they converge to a neighbourhood of the QRE, whose size decreases with increasing $T$}
            \label{fig::Conflict_Network}
    \end{figure*}
    
    We now examine competitive games who do not satisfy the network zero sum assumption. To do this we consider \emph{conflict networks}
    as considered in \cite{ewerhart:fp}, which cover a wide array of competitive games including the widely studied Colenol Blotto
    game \cite{roberson:blotto}. Strictly speaking, a conflict network is one in which the pairwise bimatrix game $(A^{kl}, A^{lk})$
    satisfies
    \begin{align*}
        (A^{kl})_{ij} &= v^k (P^{kl})_{ij} - c^{kl}_i \\
        (A^{lk})_{ji} &= v^j (P^{lk})_{ji} - c^{lk}_j \\
    \end{align*}
    
    where $v_k, v_l > 0$, $c_{kl} \in \R^{n_k}, c_{lk} \in \R^{n_l}$ and $(P^{kl})_{ij} + (P^{lk})_{ji} = 1$ for all $i \in
    S_k, j \in S_l$. As pointed out in \cite{ewerhart:fp}, if $v_k = v_l$ for all $(k, l) \in \edgeset$, the conflict
    network is equivalent to a network zero sum game. Therefore, we consider the more interesting case where $v_k \neq
    v_l$. We ensuring that these conditions are satisfied, we generate a conflict network game, which we denote $\Gamma_C$
    of three agents. The network is fully connected and, for each agent $k$
    \begin{align*}
        A^{k, k+1} = \begin{pmatrix}
            2.4 & 6.6 \\ 4.5 & 3.1
        \end{pmatrix}, \; \; A^{k, k-1} = \begin{pmatrix}
            2.8 & 1.0 \\ 4.2 & 7.2
        \end{pmatrix},
    \end{align*} 
    and the sums $k + 1$, $k - 1$ are taken $\mod N$. As this game does not satisfy the network zero sum
     assumption, it is not necessary that (\ref{eqn::QLD}) converges to a QRE, as shown in our experiments in Figure \ref
     {fig::Conflict_Network}. By applying the procedure in Section \ref{sub:finding_the_closest_nzsg}, we find the nearest
     network zero sum game which we call $\Gamma_Z$. The payoff matrices for $\Gamma_Z$ are given in the Appendix. Next, by
     using Corollary \ref{corr::2-norm}, it is possible to show that $d(\Gamma_C, \Gamma_Z) \leq 7.2$. With this and the
     fact that Q-Learning converges in $\Gamma_Z$ \cite{piliouras:zerosum}, it is possible to use Theorem \ref
     {thm::NZSGConv} to determine the approximate convergence of (\ref{eqn::QLD}) in $\Gamma_C$. For low values of $T$, the
     region to which Q-Learning converges is large, and takes up the entire simplex. However, this region becomes smaller
     as $T$ is increased, so that Q-learning converges to a smaller neighbourhood close to the QRE of $\Gamma_Z$. 
    
    % paragraph approximate_behaviour_of_a_conflict_network (end)
    % subsection experiments_on_near_nzsg (end)
    
    \section{Conclusion} % (fold)
    \label{sec:conclusion}
    
    In this paper we begin developing an understanding of the smooth Q-Learning dynamics beyond strictly competitive many player games. We show that in
    games which are sufficiently close to satisfying the network zero-sum assumption, Q-Learning converges to within a region of a unique Quantal
    Response Equilibrium (QRE). The size of this region can be adjusted by controlling either the distance from the strictly competitive setting, or the
    exploration rates of the agents. Whilst the latter amounts to parameter tuning, we consider the former by determining a method to find, for a given
    network game, the nearest network zero-sum game (NZSG). In such a manner, the approximate behaviour of Q-Learning can be understood in arbitrary
    competitive games. In our experiments we demonstrate the utility of our results in practice, in particular showing that, even in the presence of
    noise, the asymptotic behaviour of Q-Learning can be understood in terms of distance from the QRE of an underlying NZSG. 
    
    Our results also present an avenue for extending beyond strictly cooperative settings. In particular, the approximate behaviour of Q-Learning in
    near-potential games can be examined, thus beginning to bridge the gap between strictly competitive and strictly cooperative games. Another
    interesting direction would be to extend towards \emph{weighted} NZSGs, which comprise a larger set of games than the \emph{exact} NZSG setting
    considered in this work. Finally, our method for finding the nearest NZSG requires the original game itself to be a bidirectional network game.
    Lifting this assumption would allow for the approximate behaviour of a wider class of multi-agent settings (e.g. leader-follower) games to be
    understood. 

% section conclusion (end)

\section*{Acknowledgments}
Aamal Hussain and Francesco Belardinelli are partly funded by the UKRI Centre for Doctoral Training in Safe and Trusted Artificial Intelligence (grant number EP/S023356/1). This research/project is supported in part by the National Research Foundation, Singapore and DSO National Laboratories under its AI Singapore Program (AISG Award No: AISG2-RP-2020-016), NRF 2018 Fellowship NRF-NRFF2018-07, NRF2019-NRF-ANR095 ALIAS grant, grant PIESGP-AI-2020-01, AME Programmatic Fund (Grant No.A20H6b0151) from the Agency for Science, Technology and Research (A*STAR) and Provost’s Chair Professorship grant RGEPPV2101.

\bibliographystyle{ieeetr}
\bibliography{references}

\end{document}